\PassOptionsToPackage{x11names}{xcolor}
\documentclass[acmsmall]{acmart}
\AtBeginDocument{%
  }

\setcopyright{none}
\acmArticle{}
\renewcommand\footnotetextcopyrightpermission[1]{}  %
\AtBeginDocument{\pagestyle{plain}}

\usepackage{amsthm,amsfonts}
\usepackage{url}
\usepackage[capitalize]{cleveref}
\usepackage{enumitem}
\usepackage{subcaption}
\usepackage{wrapfig}

\newcommand{\bs}{B} %
\newcommand{\spm}{S} %
\newcommand{\gp}{g} %
\newcommand{\bsmax}{\bs_\text{max}} %
 
\newcommand{\gpmin}{\gp_\text{min}} %
\newcommand{\spc}{s} %
\newcommand{\spp}{s_\text{used}} %

\newcommand{\gpu}{\bar{\gp}}

\makeatletter
\def\dt@TestPunct#1{\ifx\dt@TestPunct#1\dt@TestPunct\else\dt@Puncttrue\fi}
\def\dt@TestDot#1.\dt@TestEnd#2\dt@TestStop{\dt@TestPunct{#2}}
\def\dt@TestQuestM#1?\dt@TestEnd#2\dt@TestStop{\dt@TestPunct{#2}}
\def\dt@TestExclaM#1!\dt@TestEnd#2\dt@TestStop{\dt@TestPunct{#2}}
\def\dt@TestColonM#1:\dt@TestEnd#2\dt@TestStop{\dt@TestPunct{#2}}
\newif\ifdt@Punct
\def\dt@MaybeAddPunct#1{%
  \dt@Punctfalse%
  \dt@TestDot#1\dt@TestEnd.\dt@TestEnd\dt@TestStop%
  \dt@TestQuestM#1\dt@TestEnd?\dt@TestEnd\dt@TestStop%
  \dt@TestExclaM#1\dt@TestEnd!\dt@TestEnd\dt@TestStop%
  \dt@TestColonM#1\dt@TestEnd:\dt@TestEnd\dt@TestStop%
  #1\ifdt@Punct \else .\fi%
}

\def\dt@ignorespacesandimplicitepars{%
  \begingroup
  \catcode13=10
  \@ifnextchar\relax
    {\endgroup}%
    {\endgroup}%
}
\def\dt@ignorespacesandallpars{%
  \begingroup
  \catcode13=10
  \@ifnextchar\par
    {\endgroup\expandafter\ignorespacesandallpars\@gobble}%
    {\endgroup}%
}

\newcommand{\parheadnoperiod}[1]{\smallskip \noindent {\bfseries\boldmath\ignorespaces {#1}}\hskip 0.9em plus 0.3em minus 0.3em \dt@ignorespacesandimplicitepars}
\newcommand{\parhead}[1]{\parheadnoperiod{\dt@MaybeAddPunct{#1}}\dt@ignorespacesandimplicitepars}

\makeatother

\newtheorem{theorem}{Theorem}[section]
\newtheorem{proposition}[theorem]{Proposition}

\theoremstyle{definition}

\theoremstyle{remark}
\newtheorem{remark}[theorem]{Remark}

\title{Blockspace Under Pressure: An Analysis of Spam MEV on High-Throughput Blockchains}

\author{Wenhao Wang}
\authornote{Part of this work was conducted during the Summer 2025 Research Internship at Category Labs.}
\affiliation{%
  \institution{Yale University, IC3}
  \country{USA}
}

\author{Aditya Saraf}
\authornotemark[1]
\affiliation{%
  \institution{Cornell University}
  \country{USA}
}

\author{Lioba Heimbach}
\affiliation{%
  \institution{Category Labs}
  \country{USA}
}

\author{Kushal Babel}
\affiliation{%
  \institution{Category Labs}
  \country{USA}
}

\author{Fan Zhang}
\affiliation{%
  \institution{Yale University, IC3}
  \country{USA}
}
\authorsaddresses{}
\date{}

\begin{document}

\begin{abstract}
On high-throughput, low-fee blockchains, a qualitatively new form of maximal extractable value (MEV) has emerged: searchers submit large volumes of speculative transactions, whose profitability is resolved only at execution time.
We refer to this as \emph{spam MEV}.
On major rollups, it can at times consume more than half of block gas, even though only a small fraction of probes ultimately results in a trade.
Despite growing awareness of this phenomenon, there is no principled framework for understanding how blockchain design parameters shape its prevalence and impact.

{We develop such a framework, modeling spam transactions competing for on-chain opportunities under a competitive equilibrium that drives their profits to zero, and deriving equilibrium spam volumes as a function of block capacity, minimum gas price, and the transaction fee mechanism.
Empirical evidence from Base and Arbitrum is consistent with several predictions of the model: spam grew sharply as block capacity was scaled up and fell when minimum gas prices were introduced.
Our analysis yields three main insights.
First, spam is always costly: when block capacity is scarce, it displaces users and drives up gas prices; as block capacity grows, it increasingly consumes execution resources, raising network externality, i.e., the cost of provisioning and processing blocks.
We show that spam takes an increasing share of each additional unit of block capacity, so capping it before all users are included creates a favorable trade-off: forgoing a small amount of user welfare eliminates disproportionate spam and externality.
Second, we extend the analysis to approximate priority fee ordering and show that it can reduce spam by making earlier block positions more expensive, as spammers pay more to reach early block positions.
Third, as user demand grows and blockspace is scaled accordingly, spam's share of block capacity plateaus rather than growing indefinitely.
}

\end{abstract}



\maketitle
\thispagestyle{plain}

\section{Introduction}
\label{sec:intro}

Maximal extractable value (MEV) has been a central topic in blockchain research over the past several years~\cite{daian2020flash,babel2023clockwork,qin2021quantifying}. Much of this work has focused on the Ethereum Layer~1, which, for a long time, was the primary home of decentralized finance (DeFi) and thus the main arena for MEV extraction. On the Ethereum Layer~1, the most widely studied MEV strategies, such as sandwich attacks, cyclic arbitrage, liquidations, and non-atomic arbitrage, are precise and targeted: a searcher identifies a specific on-chain opportunity through off-chain computation, constructs a transaction to capture it, and submits it with high confidence of success. These strategies rely on the ability to observe pending transactions in a mempool, i.e., the public waiting area for transactions awaiting inclusion, or to monitor state changes with sufficient time to react. Today, competition among searchers on Ethereum is mediated by block builder auctions~\cite{heimbach2023pbs}, in which searchers bid for the right to capture an opportunity. Only the winning transaction is included on-chain; losing bids are largely filtered out before execution, keeping most failed MEV attempts off the blockchain.

A qualitatively different form of MEV has recently emerged on high-throughput blockchains. Rather than targeting a specific opportunity identified off-chain, searchers submit high volumes of speculative transactions whose profitability is resolved only at execution time.
We refer to this class of strategies as \textbf{spam MEV}.
Related work has called this phenomenon \emph{optimistic MEV}~\cite{solmaz2025optimistic} or \emph{probabilistic MEV}~\cite{mazorra2026timinggames}, emphasizing the searcher's uncertainty about whether a transaction will be profitable.
Our terminology instead emphasizes that they consume shared block space and infrastructure resources, whether or not each transaction succeeds.

In spam transactions, both the detection and execution of opportunities reside largely in on-chain smart-contract logic~\cite{solmaz2025optimistic}: a transaction probes whether a profitable opportunity exists at the moment of execution and, if so, captures it, while when no opportunity is found it may revert or simply consume gas, i.e., the computational cost of executing transactions on the blockchain, searching without ever executing a trade.
This encompasses a range of speculative strategies, of which two are especially common.
In \emph{cyclic arbitrage}, a bot checks whether a profitable route across DEX pools exists and trades only if it does.
In \emph{liquidation probing}, a bot checks whether a lending position has become liquidatable and liquidates it only if so.
\Cref{app:data} describes how we detect both on-chain.

To illustrate the contrast, on Ethereum Layer~1 a cyclic arbitrage bot generally precomputes a profitable path and executes a single atomic transaction, whereas on many high-throughput chains, the same type of bot speculatively submits transactions without knowing whether an opportunity exists, repeatedly probing liquidity pools in the hope of capturing small gains and accepting a high failure rate in exchange for marginal profits.

First documented on Solana~\cite{umbra2022solana}, spam MEV has since been identified on Ethereum Layer~2 rollups, including Base, Optimism, and Arbitrum~\cite{solmaz2025optimistic,gogol2025first}. The scale of the phenomenon is striking: on Base and Optimism, spam MEV transactions consumed over 50\% of block gas in Q1~2025~\cite{solmaz2025optimistic}, yet only 6--12\% of these speculative probes result in an actual trade. On top of that, despite consuming more than half of on-chain gas, spam MEV transactions pay less than a quarter of total fees~\cite{solmaz2025optimistic}, as they carry low priority fees. %

What makes spam MEV viable and prevalent? A combination of architectural and economic features common to modern high-throughput chains likely plays a role. Low transaction fees reduce the cost of failed speculative transactions, making repeated probing profitable even at low success rates. Fast block times, often below one second on rollups, can leave insufficient time for searchers to observe state changes and submit targeted transactions before the next block is produced, favoring continuous speculative submission over deliberate off-chain computation.
Similarly, most high-throughput chains lack a public mempool, e.g., rollups use centralized sequencers and Solana forwards transactions directly to the current block producer, limiting the information available for targeted extraction. %
Designs such as encrypted mempools~\cite{agarwal2025efficiently,babel2024prof, fernando2025trx}, when deployed, are likely to only increase spam MEV by further reducing the information available for targeted MEV strategies.

While the prevalence of spam MEV has been documented, there is no principled framework to study the interplay between design parameters and spam volumes, and the impact of spam on various stakeholders.
As a result, the response from various blockchain designers has been largely reactive and ad-hoc (c.f. Section~\ref{sec:ecosystem-response}).

In this work, we fill this gap by developing a framework to determine the equilibrium volume of spam as a result of block space and fee design choices, and to analyze the impact of resulting spam on user welfare, validator revenue, and network externality, i.e., the cost of provisioning block capacity and processing transactions.%
We characterize the tradeoffs involved in design choices and provide guidance for blockchain designers to navigate these tradeoffs.

In the remainder of this section, we first highlight the industry's response to spam MEV, describe the design levers and metrics we focus on in this work, and then summarize our contributions.

\subsection{Ecosystem Response}
\label{sec:ecosystem-response}

The prevalence of spam MEV has not gone unnoticed. The blockchain community has become increasingly aware of the scale of the problem, and a debate has emerged about how harmful spam truly is and what should be done about it. Flashbots' thesis that MEV fundamentally limits scaling~\cite{miller-mev-scaling-2025}, i.e., that added throughput capacity is absorbed by spam rather than passed on to genuine users, has been particularly influential in framing the discussion.

The most visible case study is Base, the Ethereum Layer~2 rollup with the highest DeFi total value locked (TVL)~\cite{defillama-rollups-2025}. Following the Dencun upgrade, which drastically reduced Layer~2 data availability costs, Base increased its block gas limit. However, the additional capacity was largely consumed by spam MEV rather than benefiting genuine users~\cite{solmaz2025optimistic,gogol2025first}. As \Cref{fig:base-spam-growth} shows, spam gas grew disproportionately faster than non-spam gas as Base progressively raised its gas target after Dencun, absorbing about half of the capacity added through early 2025. When the community took notice~\cite{solmaz2025optimistic,gogol2025first,miller-mev-scaling-2025}, Base reduced the gas target from 70M to 50M. \Cref{fig:base-target-decrease} shows that comparing the 30 days before and after the change, spam gas fell by 34\% and non-spam gas by 24\%, i.e., both declined but spam absorbed a larger share of the reduction. Beyond block space consumption, spam also imposes costs on the broader network infrastructure. Multiple blockchain indexing services, including Etherscan, announced they would discontinue free API access for Base, citing the significant infrastructure costs driven by the high transaction volumes~\cite{lefterisjp-base-indexing-2025, ijaack94-indexer-base-2025, etherscan-base-api-2024}. Base ultimately raised its minimum transaction fee, explicitly citing spam reduction as the rationale~\cite{base-fee-increase-2025}. As \Cref{fig:base-spam-decline} shows, spam gas on Base trended downward thereafter, coinciding with successive increases in the protocol minimum gas price, though other factors may have also contributed.

\begin{figure}[h!]
    \centering
    \begin{subfigure}[t]{0.49\linewidth}
        \centering
        \includegraphics{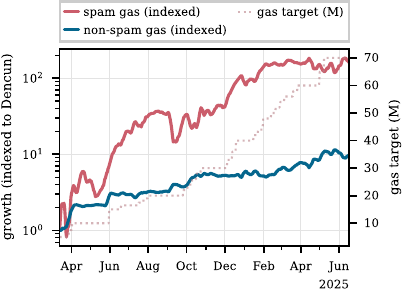}
        \caption{Spam and non-spam gas on Base indexed to Dencun (March~13, 2024). Spam gas grew 184$\times$ from Dencun to its peak (June~2025), while non-spam gas grew 11.4$\times$.}
        \label{fig:base-spam-growth}
    \end{subfigure}
    \hfill
    \begin{subfigure}[t]{0.49\linewidth}
        \centering
        \includegraphics{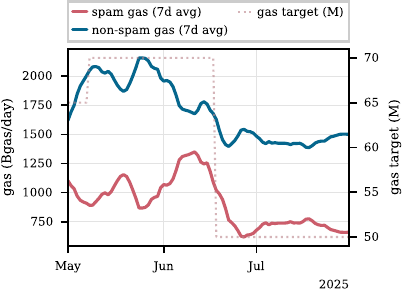}
        \caption{Spam and non-spam gas on Base around the gas target reduction from 70M to 50M on June~18, 2025. In the 30 days after the change, spam gas fell by 34\% and non-spam gas by 24\%; spam's share of total gas dropped from 37\% to 34\%.}
        \label{fig:base-target-decrease}
    \end{subfigure}
    \caption{Spam and non-spam gas on Base over time (7-day moving averages).}
    \Description{Two time-series panels compare spam and non-spam gas on Base. In the left panel, indexed gas use is on a logarithmic vertical axis: the spam curve rises much faster than the non-spam curve as the gas target increases. In the right panel, both daily gas-use curves decline around the June 2025 gas-target reduction, with a sharper proportional decline in spam. Dotted step lines show the gas target on the right axes.}
    \label{fig:base-throughput}
\end{figure}

Other chains have moved in the same direction. Arbitrum, the rollup with the second-highest DeFi TVL~\cite{defillama-rollups-2025}, had set a minimum gas price meaningfully above 1~wei well before spam attracted broader attention, and recently doubled this floor as part of a broader fee mechanism overhaul~\cite{arb-fee-tweet-2025, arbitrum-dia-2026}. Outside the EVM ecosystem, Aptos similarly raised its minimum gas fee, citing spam reduction~\cite{aptos-gas-fee-2025}.
Newer chains have launched with spam in mind: Monad, for instance, sets a non-trivial minimum gas price and charges based on the gas limit rather than gas consumed, making spam MEV with large unused gas allocations more costly~\cite{categorylabs-monad-spec-2025}.

\begin{figure}[t]
    \centering
    \includegraphics{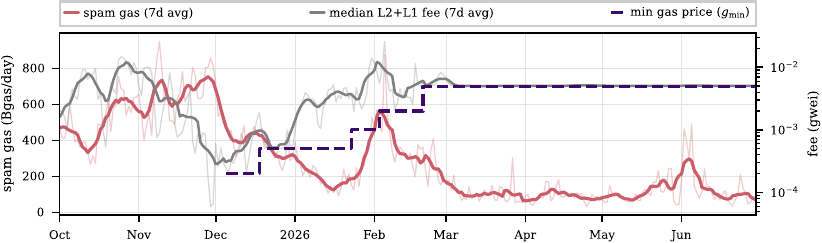}
    \caption{Daily spam gas on Base from October~2025 onward (7-day moving average), alongside the median fee and protocol minimum gas price (log scale, right axis). Average spam gas fell from 495~Bgas/day in October~2025 to 355~Bgas/day in February~2026 ($-28\%$).}
    \Description{A time-series chart plots daily spam gas on the left axis and fees on a logarithmic right axis. The thick spam curve generally declines from late 2025 into 2026, with temporary rebounds in February and June. A dashed step line shows successive increases in the minimum gas price, and a gray curve shows the median total fee. Thin lines show daily variation around the moving averages.}
    \label{fig:base-spam-decline}
\end{figure}

These interventions share a common intuition: spam is an unwanted form of resource abuse that reduces useful throughput and imposes costs on the network.
Yet the responses have been largely reactive and chain-specific, without formal analysis of how design parameters interact to determine the equilibrium level of spam.
A key factor is that, unlike Ethereum~L1 where congestion pricing naturally emerges, high-throughput chains can have slack capacity where the market-clearing gas price approaches zero, while the network still bears the cost of processing every transaction.
A minimum gas price $\gpmin$ ensures that each transaction pays a marginal cost and, as our analysis shows, serves as a key lever against spam in this regime.

\subsection{Design Levers}
\label{subsec:design-levers}

We now describe the blockchain design choices that influence the presence and characteristics of spam transactions. These are the parameters that our framework analyzes.

\parhead{Block Space Limit $\bsmax$}
The block space limit affects inclusion costs.
When the limit $\bsmax$ is low relative to transaction demand, competition for inclusion in the block intensifies, increasing gas prices $\gp$ and may therefore reduce the expected revenue of each spam transaction.
Conversely, when $\bsmax$ is large enough to accommodate all transactions, gas prices $\gp$ decrease, and therefore increase the expected profit of spam transactions.

\parhead{Minimum Gas Price $\gpmin$}
A blockchain can impose a minimum gas price $\gpmin$ per transaction, independent of block space availability. Higher values of $\gpmin$ increase the cost when block capacity is slack, influencing the economic outcome for both spam and non-spam transactions. Without a gas price floor, the equilibrium gas price $\gp$ would approach zero when block capacity exceeds demand, resulting in a disproportionately high volume of spam transactions.

\parhead{Transaction Fee Mechanism}
The design of the transaction fee mechanism (TFM) impacts transaction submission behavior.
In a blockchain, the TFM determines which transactions are included, how much they pay, and, when applicable, how they are ordered within the block.
For instance, spam transactions often specify a large gas limit but typically consume only a fraction of it when executed without profitable arbitrage opportunities~\cite{solmaz2025optimistic}.
Therefore, a TFM that imposes higher costs for specifying large gas limits can change the spam MEV strategy.
Additionally, a TFM can dictate the ordering of transactions within a block, such as through priority-based ordering, which affects the resulting spam volumes.

\subsection{Our Framework and Contributions}
\label{subsec:contributions}

We develop a framework for analyzing spam MEV and its interaction with blockchain design choices.
We model spam transactions as competing for on-chain opportunities, and study how equilibrium spam volume depends on the block space limit $\bsmax$, the minimum gas price $\gpmin$, and the transaction fee mechanism.
We evaluate the impact of spam on user welfare, validator revenue, and network externality.
Our main results are:

\begin{enumerate}[leftmargin=*]
    \item \textbf{Closed-form equilibrium spam volume and gas prices} (\cref{subsec:equilibrium-random}): We characterize spam equilibrium under random transaction ordering. Depending on block capacity and the minimum gas price, the outcome is either no spam entry, a slack equilibrium pinned down by $\gpmin$, or a congested equilibrium in which spam raises the gas price.

    \item \textbf{User welfare, revenue and externality characterization along with trade-offs} (\cref{subsec:welfare-analysis}): We analyze user welfare, validator revenue, and network externality in equilibrium.

    \item \textbf{Parameter-setting guidance} (\cref{subsec:parameter-guidance}): We characterize the choice of $(\bsmax, \gpmin)$, and show how to use the proportion of spam at the margins as a practical rule for selecting the design parameters without admitting disproportionate spam.

    \item \textbf{Priority fee ordering (PFO)} (\cref{sec:priority-fee-ordering}): We extend the analysis to an approximate PFO TFM. We show that PFO can reduce spam by making earlier positions in the block more expensive. %

    \item \textbf{Demand-scaling analysis} (\cref{sec:demand-scaling} and \cref{sec:empirical}): We study whether spam limits blockchain scaling~\cite{miller-mev-scaling-2025}. 
    Using empirical data, we estimate that MEV opportunity size grows at least linearly with non-MEV trading volume (near-linearly on Base, super-linearly on Arbitrum).
    We show that as demand grows, spam remains a nontrivial share of included gas, but its share plateaus rather than growing without bound.
    
    \item \textbf{Empirical evidence and case studies} (\cref{sec:empirical}): We present case studies of spam on Base and Arbitrum. Our findings support several model predictions: spam falls when gas price floors rise, grows with block capacity, absorbs disproportionate marginal block space, and, on Base, appears later in the block as predicted by the PFO analysis.

    \item \textbf{Mitigations} (\cref{sec:mitigations}): We discuss mitigations with the insights from our model and analysis.

\end{enumerate}

\section{Model}
\label{sec:modeling}

We present our model for analyzing spam MEV.
The design levers (block space limit $\bsmax$, minimum gas price $\gpmin$, and transaction fee mechanism) were introduced in \cref{subsec:design-levers}.
We consider blockchains run by validators; we refer to the parties that send transactions as users.
There are two types of users in the model: genuine users and spam bots.
Genuine users submit transactions because they receive intrinsic utility from execution.
Spam bots, by contrast, are strategic agents that submit transactions in order to search and capture MEV opportunities that are determined at the time of execution. They have no intrinsic utility, and their transaction volume is determined endogenously in equilibrium.

We analyze TFMs that charge a transaction for its gas limit rather than gas used, and block size is counted in terms of transaction gas limits, similar to high-throughput blockchains such as Solana and Monad~\cite{categorylabs-monad-spec-2025}. 
Our results can be straightforwardly adapted to a mechanism that charges only for the used gas.
For our initial analysis in Section~\ref{sec:random-ordering}, the TFM orders transactions in a random order, mimicking the operations of a low-latency blockchain. We extend our analysis to priority-fee based ordering in Section~\ref{sec:priority-fee-ordering}.

Now we formalize block space and spammers' utility model and define the metrics we use to evaluate equilibrium outcomes.

\subsection{Block Demand and Spam Model}
\label{subsec:utility-model}

We analyze transaction dynamics within a single block.
Let $\bs\le \bsmax$ denote the amount of block space actually occupied in the block.
Let $Q_u\le \bs$ represent the amount of included user gas, i.e., gas consumed by genuine (non-spam) transactions.
The demand from genuine users follows a demand function $D(\gp)$, which specifies how much genuine-user gas is demanded at any given gas price $\gp$.
This demand curve does not include spam demand.
Each spam transaction carries a gas limit of $\spc$, so if $\spm$ spam transactions are included, occupied block space decomposes as
$
\bs = Q_u+\spc\spm.
$
Equivalently, $\spm=(\bs-Q_u)/\spc$ in terms of $\bs$ and $Q_u$.

In the TFM model used in \cref{sec:random-ordering}, the TFM gives a clearing price $\gp^* \geq \gpmin$ that is paid by all included transactions.
In the PFO model in \cref{sec:priority-fee-ordering}, the same logic applies sub-block by sub-block, with different clearing prices for different execution positions.
All included transactions, including genuine user transactions that create opportunities and spam transactions that try to capture them, are executed in the order determined by the TFM.
Let $\mathcal{R}$ denote the set of arbitrage opportunities created during this ordered execution, and let $r_\ell$ be the value of opportunity $\ell\in\mathcal{R}$.
The total opportunity value in the block is
$
r:=\sum_{\ell\in\mathcal{R}} r_\ell.
$
A spam transaction can capture the opportunities that have appeared earlier in the execution order and have not already been captured by a previous spam transaction.
Equivalently, each opportunity is assigned to the first spam transaction sequenced after the genuine user transaction that creates it; if no later spam transaction appears, that opportunity remains unclaimed.
We model the total opportunity size $r$ to be increasing in the amount of included user gas.

We assume a free-entry environment for spam bots: any bot can submit spam transactions.
There is no fixed cost of becoming a spammer, and the spammer only pays for the transaction fees of submitted spam transactions.
Strategic spam bots enter as long as their expected net payoff is non-negative.
At equilibrium, included spam transactions earn zero expected net profit (full rent dissipation~\cite{hillmansamet1987,posner1975}), and an additional spam transaction cannot profitably enter.
Given a TFM and an ordering rule, the clearing price and the amount of spam are endogenously determined by genuine user demand function and this zero-profit entry condition.

\begin{remark}[Symmetric searchers]
Our equilibrium treats spam transactions within the same execution region as symmetric, i.e., conditional on inclusion, they have the same expected ability to capture opportunities under the modeled ordering rule.
The assumption is a natural baseline for speculative spam MEV, where transactions are submitted before a specific opportunity is observed and opportunity discovery occurs during execution, so observation latency is less central than in targeted MEV.
In practice, latency, strategy, and execution infrastructure may still affect inclusion and ordering, and searchers with infrastructure or strategy advantages may retain positive profit.
We note that modeling these differences is outside the scope of this work.
\end{remark}

The utility of a spam transaction is its realized MEV revenue minus the fees it pays.
Its realized MEV revenue is the total value of all opportunities for which it is the first spam transaction sequenced after the opportunity is created.
In the single-price model, if the transaction pays clearing price $\gp^*$, its utility is therefore
$
\sum_{\ell\in\mathcal{R}_{\mathrm{cap}}} r_\ell-\spc\cdot\gp^*,
$
where $\mathcal{R}_{\mathrm{cap}}$ is the set of opportunities captured by that spam transaction.
In the PFO model, the fee term uses the clearing price of the sub-block in which the spam transaction is included.
The utility of a genuine user transaction is defined as $u-f$ if included in the block, and $0$ otherwise, where $u$ is its valuation of execution and $f$ is the fee it pays.
For included genuine users, $f\le u$.
Looking ahead, user valuations are represented by the demand function, while fees are determined by the TFM.

{\begin{remark}[Modeling choices]
\label{rem:modeling-choices}
Our model is built to isolate the effect of block capacity, minimum gas prices, and ordering mechanisms on spam volume.
We use linear demand to obtain closed-form equilibrium expressions.
We note that the same analysis methods extend to other decreasing demand curves, and we discuss the exponential demand case in \cref{app:exp-demand}.

Random ordering is used as a baseline because it captures the idea that, on low-latency chains, a spam transaction may obtain an early execution position without paying a high priority fee.
\Cref{sec:priority-fee-ordering} then studies how the results change when earlier positions are priced.
Finally, spam transactions are charged by gas limit in the main analysis, which matches chains where reserved execution capacity is the relevant scarce resource.
The analysis of when transactions are charged based on the gas used is discussed in \cref{app:gas-used-charging}.
\end{remark}}

\subsection{Metrics}
\label{subsec:metrics}
We now describe the metrics used to quantify the impact of spam transactions on the blockchain.
In addition to the welfare and cost metrics below, we analyze the {\em spam share} of included gas, defined as the fraction of included gas consumed by spam transactions.
A higher spam share means that a larger fraction of the block is unavailable to genuine users.
It is therefore a direct measure of the resource-abuse aspect of spam MEV.

\begin{wrapfigure}{r}{0.5\linewidth}
\centering
\includegraphics[width=179.59923pt,height=107.13942pt]{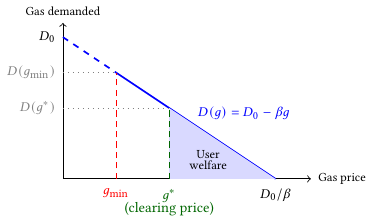}
\caption{Linear demand curve $D(\gp) = D_0 - \beta \gp$ for genuine users. At clearing price $\gp^*$, the shaded triangle is the user welfare (aggregate surplus of included users).}
\Description{A downward-sloping straight demand curve has gas price on the horizontal axis and gas demanded on the vertical axis. Vertical guides mark the minimum gas price and a higher clearing price; horizontal guides mark the corresponding demand levels. The triangular area under the curve to the right of the clearing price is shaded to represent user welfare.}
\label{fig:demand-curve}
\end{wrapfigure}

\parhead{User welfare}
The user welfare is defined as the aggregate utility of all genuine users.
As shown in \cref{fig:demand-curve} for a linear demand curve, the user welfare is the area under the genuine-user demand curve above the clearing price $\gp^*$.
A higher user welfare indicates a greater overall benefit to users of the blockchain.
Spam can reduce user welfare by reducing block space for genuine users, or by increasing the inclusion price and pricing out users with lower valuations.

\parhead{Validator revenue}
Validator revenue is the total gas fees collected from all included transactions (both users and spammers), written as
$R = \gp^* \cdot \bs$,
where $\gp^*$ is the equilibrium clearing gas price and $\bs$ is total gas used.
This quantity represents the direct income to the network from transaction processing.
For the purposes of this work, we do not consider burning of fees or block rewards, as they are merely a way to change allocation of revenue between token holders and validators, and do not affect the net economic incentives for spamming or the overall welfare analysis.

\parhead{Network externality}
The network externality is the cost borne by the network for provisioning block capacity and processing transactions, written as
$
E(\bsmax) = c_1  \bsmax + c_2 \bs,
$
where $c_1$ is the per-unit cost of providing block capacity (e.g., bandwidth provisioning, data availability, node hardware), and $c_2$ is the per-unit cost of computation (e.g., execution, state access).
The capacity component $c_1  \bsmax$ depends on the provisioned block capacity, not actual usage, while the compute component $c_2  \bs$ scales with actual gas consumed.
Network externality captures costs such as increased barriers to decentralization and additional pressure on infrastructure operators (e.g., full nodes), which are not easily quantifiable but are important considerations in blockchain design, especially as it relates to spam transactions.
In this sense, externality captures the security and reliability cost of allowing high-volume traffic that consumes shared blockchain resources.
We note that our externality measure does not include private investments by searchers in improving their infrastructure or strategy.
\footnote{Buterin~\cite{buterin2018resource}, drawing analogy to pollution, puts forth a similar approach to modeling the externality of block sizes on the ecosystem, where the externality is linear in block size for modern blockchains.}

\section{Analysis with Random Ordering}
\label{sec:random-ordering}

In this section, we study a simple random-ordering TFM, which is intentionally minimal.
There is a block of capacity $\bsmax$, and a collection of on-chain opportunities that spam transactions try to claim.
Each spam transaction reserves $\spc$ gas and is charged based on its gas \emph{limit}, not on ex post gas used.
In this section, we use the simple linear demand function ($D(\gp)=D_0-\beta \gp$) for non-spam users' gas consumption so that we can have closed-form equilibrium expressions.
We include results under exponential demand in \cref{app:exp-demand}.
Recall that $Q_u$ denotes the amount of user gas included in the block, and $r$ denotes the aggregate value of all opportunities in the block.
We assume that this aggregate opportunity size is linear in the included user gas given by $r = \frac{Q_u}{D_0} \cdot r_0$.
Finally, we impose a floor $\gpmin$, which lets us study the design choice of a minimum inclusion price.

\subsection{Spam Volume at Equilibrium}
\label{subsec:equilibrium-random}

The first quantity we need is the expected MEV revenue of each spam transaction.
We start with the case of a single opportunity, and then show that multiple non-interacting opportunities reduce to the same expression with $r$ equal to their aggregate value.

Suppose first that there is a single opportunity of value $\hat r$ and that there are $\spm$ spam transactions in the block.
Only the relative position of the opportunity among these $\spm$ spam transactions matters.
For example, when $\spm=3$, we can visualize the situation by placing the three spam
transactions as triangles and considering the four possible slots of the opportunity:
\[
\raisebox{-12.03917pt}[12.03917pt][12.03917pt]{\includegraphics[width=182.93039pt,height=24.07834pt]{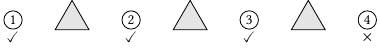}}
\]
The circles represent the possible positions where the opportunity can appear relative to the spam transactions.
In the first three positions, there is at least one spam transaction after the opportunity, so a spam transaction can claim it.
In the last position, the opportunity appears after all spam transactions, so it cannot be claimed.
More generally, there are $\spm+1$ equally likely relative positions for the opportunity, and exactly one of them fails.
Therefore,
$$
\Pr[\text{opportunity claimed by spam}]=\frac{\spm}{\spm+1}.
$$
The expected total value captured from this single opportunity is therefore $\hat r \cdot \frac{\spm}{\spm+1}$.
Since spam transactions are symmetric, the expected revenue of each spam transaction is
$
\frac{\hat r}{\spm+1}
$.

Now suppose that there are multiple opportunities in the block, indexed by $\ell=1,\dots,m$, with values $r_\ell$.
We assume these opportunities are non-interacting: their values add, capturing one opportunity does not change the value or feasibility of another, and each spam transaction can attempt to claim any opportunity that appears before it.
For each opportunity $\ell$, the same argument of expected revenue captured by each spam transaction applies.
By linearity of expectation, the expected aggregate value captured by spam is
$
\frac{\spm}{\spm+1}\sum_{\ell=1}^{m}r_\ell.
$
By symmetry across spam transactions, the expected revenue of each spam transaction is therefore
$
\frac{1}{\spm+1}\sum_{\ell=1}^{m}r_\ell
$.
Therefore, for the random-ordering analysis, multiple non-interacting opportunities are equivalent to a single aggregate opportunity of value
$
r:=\sum_{\ell=1}^{m}r_\ell
$.

We next compute the clearing price at a fixed spam volume.
Recall that the user demand curve is characterized by $D = D_0 - \beta \cdot \gp$ where $\gp$ is the gas price.
In the absence of spam, the market-clearing price implied by the demand curve and block size $\bsmax$ is
$\gp_1 := \frac{D_0 - \bsmax}{\beta}$.
However, since the protocol enforces a price floor $\gpmin$, the actual clearing price without spam is $\hat{\gp} :=\max\{\gpmin,\gp_1\}$.
If $\spm$ spam transactions are included in the block and each reserves $\spc$ gas, then the remaining space available to non-spam users is $\bsmax - \spm \spc$.
The resulting clearing price is
$\gp(\spm)=\max\left\{\gpmin,\gp_1+\frac{\spc}{\beta}\spm\right\}$.
The term $\frac{\spc}{\beta}\spm$ captures the price increase caused by spam taking up block space, and the outer maximum applies the gas price floor.

Given this price, the expected utility of each spam transaction at spam volume $\spm$ is
$
u(\spm)=\frac{1}{\spm+1}r- \spc \gp(\spm).
$
The first term is the expected MEV revenue of each spam transaction, where $r$ is the aggregate value of all opportunities in the block.
The second term is the inclusion cost.
Recall that in a competitive equilibrium, spammers enter until the expected utility of each spam transaction is $0$ (see \cref{subsec:utility-model}).
Solving the zero-profit condition ($u(\spm) = 0$) gives
\[
\spm^\ast=
\begin{cases}
0, & r_0\frac{Q_u^0}{D_0}\le \spc \hat{\gp},\\
 \frac{r_0 D(\gpmin)}{D_0 \spc \gpmin}-1, &  r_0\frac{Q_u^0}{D_0}>\spc \hat{\gp} \text{ and }\hat{\gp} = \gpmin,\\
 \frac{\sqrt{(\spc-\Delta+\beta r_0/D_0)^2+4\beta r_0}-(\spc+\Delta+\beta r_0/ D_0)}{2\spc}, & \text{otherwise},
\end{cases}
\]
where $\Delta=D_0-\bsmax$, $\hat{\gp}=\max\{\gpmin,\gp_1\}$ is the clearing price in the absence of spam, and
$
Q_u^0=\min\!\left\{\bsmax,\;D(\gpmin)\right\}
$
is the amount of user gas included in the absence of spam.

The first case says that spam does not enter when the effective aggregate opportunity value in the spam-free world is below the per-transaction inclusion cost.
The second case is the slack-at-the-floor regime, where the gas price remains $\gpmin$, and the block is not congested.
In this regime, spam depends on $r_0$, $\gpmin$, and the amount of non-spam demand that remains at the floor.
The third case is the congested regime, where spam raises the clearing price above the floor.
In this regime, spam depends on the baseline opportunity parameter $r_0$, the block size, and user demand.

\begin{remark}[Connection to Mazorra et al.]
\label{rem:mazorra}
In the slack case, $\spm^\ast = r/(\spc \gpmin) - 1 = 1/c - 1$ where $c = (D_0 \spc \gpmin)/(r_0\cdot D(\gpmin))$ is the cost-to-reward ratio.
This matches the lower bound on total spam from the timing game model of Mazorra et al.~\cite{mazorra2026timinggames}, who show that equilibrium spam in an $n$-player timing game lies in $[1/c-1, 1/c]$ and converges to $1/c-1$ as $n\to\infty$.
Our competitive equilibrium corresponds to this large-$n$ limit derived using a completely different approach.
\end{remark}

For our welfare analysis later, we now derive the minimum blockspace, denoted by $B_{\mathrm{plat}}$, needed for the clearing gas price to remain at the floor $\gpmin$.
Looking ahead, the benefit of provisioning blockspace plateaus out beyond $B_{\mathrm{plat}}$.
$B_{\mathrm{plat}}$ can be written as
$
B_{\mathrm{plat}}
=
D(\gpmin)+\left(\frac{r_0 D(\gpmin)}{D_0 \gpmin}-\spc\right)_+.
$
In this equation, the first term is the amount of non-spam demand that remains at the floor price.
The second term is the equilibrium spam volume at the floor, converted into gas units.
Therefore, $B_{\mathrm{plat}}$ is the amount of block space needed to fit both user demand at $\gpmin$ and the spam that is still profitable at that price.
For $\bsmax<B_{\mathrm{plat}}$, the block is still effectively scarce, so adding capacity lowers the clearing price and increases the spam volume.
Once $\bsmax$ reaches $B_{\mathrm{plat}}$, the gas price settles at the $\gpmin$.
At that point, spam equilibrium volume reaches its peak level and no longer changes with further increases in block capacity $\bsmax$.
\Cref{fig:spam-volume} shows this pattern.
\begin{figure}[h]
    \centering
    \includegraphics{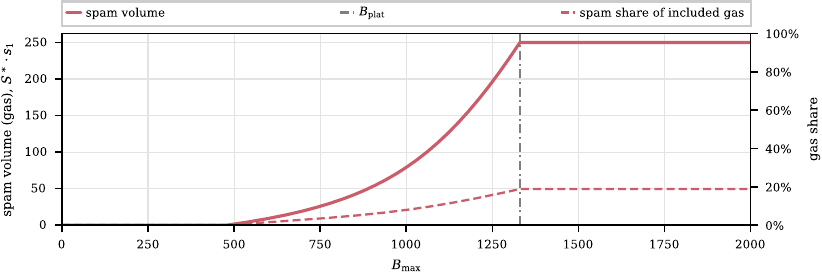}
    \caption{Equilibrium spam volume as a function of block size $\bsmax$, and the gas share of spam transactions out of total included gas.
    Spam is $0$ when the block is small and the clearing price is high.
    As $\bsmax$ grows, entry becomes more profitable and spam begins to rise.
    Once $\bsmax$ reaches $B_{\mathrm{plat}}$, the gas price is pinned at $\gpmin$ and spam plateaus.
    In this figure, $D_0=1200$, $\beta=6$, $\spc=20$, $r_0=6000$, and $\gpmin=20$.
    }
    \Description{A chart plots block size against spam gas on the left axis and spam's share of included gas on the right axis. The solid volume curve and dashed share curve are zero for small blocks, then rise until a vertical marker at the plateau threshold. Beyond that threshold, both curves are horizontal.}
    \label{fig:spam-volume}
\end{figure}

\subsection{Welfare, Revenue, and Externality Analysis}
\label{subsec:welfare-analysis}

We next study how the metrics in \cref{subsec:metrics} vary with the design parameters.
To assess the impact of spam, we compare two worlds evaluated at the same parameter values: (1) the realized world with spam, and (2) a counterfactual world without spam.
Holding parameters fixed isolates the effect of spam on users, validators, and the cost of the blockchain to process blocks.
Our analysis here, especially on welfare and externality, provides the quantities needed later for choosing blockchain parameters wisely.

Figure~\ref{fig:welfare-levels} summarizes the equilibrium levels of user welfare, validator revenue, and externality in the spam world and in the spam-free counterfactual world.
For compactness, the maths expressions of the user welfare, validator revenue, and externality, and the explicit counterfactual deltas and the characterization of where the welfare gap is largest are deferred to \cref{app:random-ordering-metric}.

\begin{figure*}[h]
    \centering
    \includegraphics{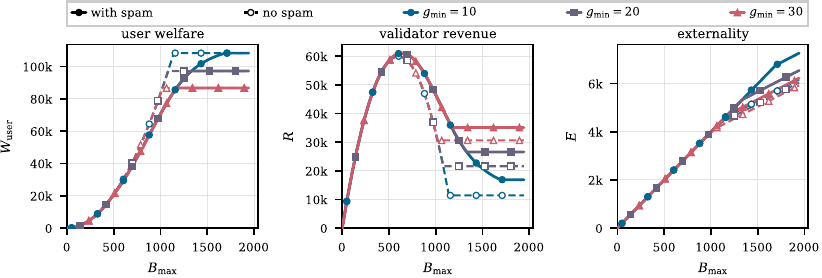}
    \caption{Levels of user welfare, validator revenue, and externality, with and without spam, as functions of block size. Each panel compares the spam world to the spam-free counterfactual at the same $(\bsmax,\gpmin)$.
    }
    \Description{Three panels plot user welfare, validator revenue, and externality against block size for three minimum gas prices. Solid curves with filled markers show outcomes with spam; dashed curves with hollow markers show outcomes without spam. Welfare rises to a plateau, revenue rises to a peak then falls toward a plateau, and externality increases. Spam delays welfare gains, raises revenue after entry, and increases externality once the no-spam world has unused capacity.}
    \label{fig:welfare-levels}
\end{figure*}

\subsection{Parameter-Setting Guidance}
\label{subsec:parameter-guidance}

The analysis above gives the system designer two main parameters against spam, namely the block size $\bsmax$ and the price floor $\gpmin$.
These two parameters affect user inclusion, validator revenue, and the system cost of supporting larger blocks.
A good choice should therefore increase user welfare, but should not keep increasing $\bsmax$ once the extra capacity mostly creates spam or idle space. Figures~\ref{fig:welfare-levels} and \ref{fig:welfare-deltas} depict the tradeoff between welfare on one side, and revenue and externality on the other side.
We address each parameter in turn: here we first discuss choosing $\bsmax$ for a fixed $\gpmin$, and we discuss choosing $\gpmin$ for a fixed $\bsmax$ in \cref{app:param-gmin}.

We present two rules for choosing $\bsmax$ given a fixed gas price floor $\gpmin$: a simple baseline that scales to the point where user welfare plateaus, and a refined rule that stops earlier by requiring that marginal capacity primarily serves users.

\parhead{A baseline rule: stop scaling at the plateau $B_{\mathrm{plat}}$.}
Choose a gas price floor $\gpmin$.
For this fixed floor, user welfare increases with $\bsmax$ until the system reaches the slack threshold $B_{\mathrm{plat}}(\gpmin)$.
At this point, the block has enough space to fit both the non-spam demand that remains at the floor and the spam that is still profitable at that floor.
Once $\bsmax \ge B_{\mathrm{plat}}(\gpmin)$, further increase in block size does not increase user welfare, because the user side is already capped by $D(\gpmin)$.
The revenue also does not increase, and only the externality increases.
Therefore, if the designer fixes $\gpmin$ and wants to maximize user welfare without creating unnecessary externality, the natural choice is
$
\bsmax^\dagger(\gpmin)=B_{\mathrm{plat}}(\gpmin).
$
This gives a relatively simple benchmark rule: for a fixed gas price floor $\gpmin$, increase block size until the user welfare plateaus.\footnote{In practice, the designer may want to keep some slack in blockspace, to accommodate bursts in user demand, i.e., demand function steepening in our model. Some blockchains do this by choosing a block capacity larger than the demand at $\gpmin$, so that demand spikes during volatile periods can be absorbed without congestion.}

\parhead{A refined rule: require that marginal capacity does not admit disproportionate spam.}
The benchmark rule $\bsmax^\dagger(\gpmin)=B_{\mathrm{plat}}(\gpmin)$ is useful, but it can be too permissive: near $B_{\mathrm{plat}}$, much of the newly provisioned capacity may be absorbed by spam, as shown in \cref{fig:marginal-capacity}.
Therefore, in practice, the designer may want a stricter rule: if the block size is increased by a small amount, a sufficient fraction of that extra capacity should go to useful user gas rather than to spam.

To make this precise, we define the marginal user share as
$
m_{\text{user}}
:=
\frac{\partial Q_u^\ast(\bsmax,\gpmin)}{\partial \bsmax}
$,
which measures the fraction of an increase in block capacity that becomes useful user gas.
Inside the entry-and-congested region, this derivative can be written as
$$
m_{\text{user}}
=
\frac{1}{2}
\left(
1-\frac{\bsmax-D_0+\spc+\beta r_0/D_0}
{\sqrt{(\bsmax-D_0+\spc+\beta r_0 /D_0)^2+4\beta r_0 }}
\right).
$$
This expression is strictly decreasing in $\bsmax$ as shown in \cref{prop:mmus-general}; equivalently, spam takes an increasing share of marginal capacity.
Once $\bsmax \ge \bs_{\mathrm{plat}}$, $m_{\text{user}}$ becomes zero, because extra capacity is no longer used by either users or spam.

We now introduce a design parameter $\eta\in(0,1]$, which we call the {\em minimum marginal user share} (MMUS).
It requires that at least an $\eta$ fraction of the next unit of capacity go to users.
Given a fixed $\gpmin$, the refined rule chooses
\[
\bsmax^\ast(\gpmin,\eta)
\in
\arg\max_{\bsmax}\;W_{\text{user}}(\bsmax,\gpmin)
\;\text{subject to}\;
m_{\text{user}}\ge \eta.
\]
Because $W_{\text{user}}(\bsmax,\gpmin)$ is nondecreasing in $\bsmax$, this rule picks the largest block size for which marginal capacity still primarily benefits users.
It is a stricter version of the plateau rule: $\bsmax^\ast(\gpmin,\eta)\le B_{\mathrm{plat}}(\gpmin)$ whenever $\eta>0$.
Lower $\eta$ allows more aggressive scaling, while higher $\eta$ stops scaling earlier.

\begin{proposition}
\label{prop:mmus-general}
Fix $\gpmin$ and suppose the equilibrium is in the entry-and-congested region.
Let $\gp^\ast$ denote the gas price at equilibrium.
Then, for any twice continuously differentiable and strictly decreasing demand curve $D(\gp)$, if $\gp^\ast D(\gp^\ast) D''(\gp^\ast)+2D(\gp^\ast)D'(\gp^\ast)-2\gp^\ast(D'(\gp^\ast))^2<0$ ($D'$ is $\frac{\mathrm{d} D}{\mathrm{d}g}$), (which holds for the linear demand function), then $m_{\text{user}}$ is decreasing in $\bsmax$.
\end{proposition}

The proof can be found in \cref{app:proofs}.
We note that the condition in \cref{prop:mmus-general} holds for the linear demand function, and also for a class of general decreasing demand curves including the exponential demand function (i.e., $D(\gp) = D_0 \cdot e^{-\lambda \gp})$.
\Cref{app:exp-demand} gives the corresponding equilibrium and welfare results under exponential demand, while \cref{rem:mmus-general-demand} shows that the sufficient condition need not hold for every decreasing demand curve.

Figures~\ref{fig:marginal-capacity} visualize this rule.
It plots $\bsmax^\ast$ as a function of $\eta$ when $\gpmin$ is fixed, and shows the corresponding $\bsmax^\ast$ values for representative $\eta$ values.

\parhead{Takeaway (Marginal spam share)} Spam takes an increasing share of each additional unit of block capacity. Setting $\bsmax$ below $B_{\mathrm{plat}}$ or raising $\gpmin$ creates a favorable trade-off: forgoing a small amount of user welfare eliminates a disproportionate amount of spam.

\begin{figure}
    \centering
    \includegraphics{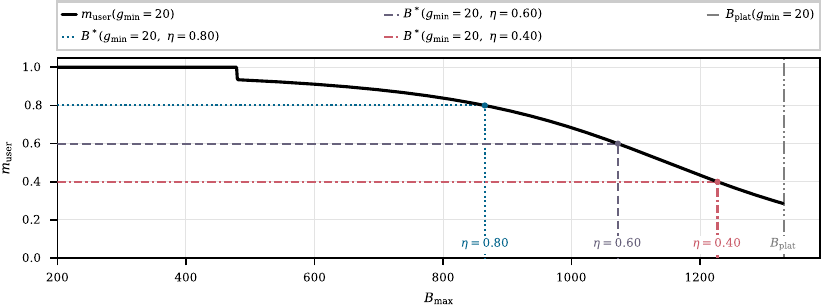}
    \caption{Marginal allocation of additional block capacity. The x-axis represents the maximum block size, and the y-axis represents the marginal user share. Here we fix $\gpmin =20$.}
    \Description{The marginal user share is plotted against maximum block size. The black curve initially equals one, drops at the onset of spam, and then decreases as capacity grows. Horizontal guides at user-share thresholds of 0.8, 0.6, and 0.4 intersect the curve; vertical guides project these intersections to the corresponding recommended block sizes. A separate vertical line marks the larger plateau threshold.}
    \label{fig:marginal-capacity}
\end{figure}

\section{Analysis with (Approximate) Priority Fee Ordering}
\label{sec:priority-fee-ordering}

We now turn to an approximate TFM with priority fee ordering (PFO), where transactions with higher bids are executed earlier.
PFO may help reduce spam, as under random ordering, spam transactions can occupy early positions in the block without paying more than others, whereas PFO makes this costly.
In this section, we adopt a PFO model which is a tractable analytical approximation intended to isolate how pricing earlier execution positions affects spam incentives.
We note that the model is not intended to reproduce the exact ordering and payment rules of any particular deployed blockchain that uses PFO.

To capture this effect, we use a block-position-specific demand model.
The idea is that users differ not only in their valuation, but also in which part of the block they are willing to occupy.
In practice, the top of the block generally serves high-value, time-sensitive flow, e.g., professional DeFi traders, while later positions are filled by users who only need inclusion.
We model this by partitioning the block into ordered regions and associating each region with its own demand curve.
High-value users only demand early positions, since by the time later positions execute their opportunity has already passed.
Lower-value users are represented by demand for later positions, which clear at lower gas prices.

In the main body, we focus on the two sub-block case.
We split the block into an early sub-block and a late sub-block, with the first executing before the second.
Transactions sort into sub-blocks by bid: higher bids land in the first sub-block, lower bids in the second.
Within each sub-block, ordering is random.
We defer the general $n$-sub-block formulation, where sub-blocks have equal capacity and the demand curve is evenly partitioned across them, to \cref{app:pfo-equilibrium}.

\parhead{Model}
We use a single parameter $v\in[0,1]$ to specify the two sub-block approximation.
The first sub-block receives fraction $v$ of block capacity and is associated with the upper $v$ fraction of the original valuation distribution; the second sub-block receives fraction $(1-v)$ of block capacity and is associated with the remaining lower $(1-v)$ fraction.
Thus, the sub-block capacities are $C_1=v\bsmax$ and $C_2=(1-v)\bsmax$.
{Formally, the demand curves are $D_1^{(v)}(\gp)=(\min\{vD_0,D_0-\beta\gp\})^+$ and $D_2^{(v)}(\gp)=\bigl((1-v)D_0-\beta\gp\bigr)^+$.}
Demand is sub-block specific: users from the first slice are willing to occupy only the first sub-block, and users from the second slice are willing to occupy only the second.
Thus, unmet demand in one sub-block does not spill over into the other, so each sub-block has its own user demand, clearing price, and spam-entry condition.

Let $\gp_1$ denote the clearing price of the first sub-block and let $\gpu$ denote the clearing price of the second sub-block.
Whenever both sub-blocks are nonempty, we have $\gp_1\ge \gpu\ge \gpmin$.
Spam bots may target either sub-block and pay the corresponding price.
If sub-block $i$ contains $Q_i$ user gas, the aggregate value of opportunities created in that sub-block is $\frac{r_0}{D_0}Q_i$.

\subsection{Competitive Spam Equilibrium}
\label{subsec:pfo-equilibrium}

We now characterize the spam equilibrium in the two-sub-block case.
Let $S_1,S_2$ denote the number of spam transactions placed in the first and second sub-blocks, respectively.

Given $S_1$ spam transactions in the first sub-block, the included user gas there is
$
Q_1(S_1;\gpu)
=
\min\{
C_1-S_1\spc,
D_1^{(v)}(\gpu)
\}^+.
$
The corresponding price is
$
\gp_1(S_1;\gpu)
=
\max\{
\gpu,
P_1^{(v)}(Q_1(S_1;\gpu))
\}
$ for the first sub-block.
Intuitively, if the top-sub-block demand at the lower price $\gpu$ is large enough to fill the post-spam capacity $C_1-S_1\spc$, then the first-sub-block price rises above $\gpu$ to clear that capacity.
Otherwise, the first sub-block is slack and clears at $\gpu$.
Similarly, given $S_2$ spam transactions in the second sub-block, the included user gas there is
$
Q_2(S_2;\gpu)
=
\min\{
C_2-S_2\spc,
D_2^{(v)}(\gpu)
\}^+.
$
The second sub-block clears at the block-level inclusion price, so the equilibrium value $\gpu^\ast$ satisfies the fixed-point condition
$
\gpu^\ast
=
\max\{
\gpmin,
P_2^{(v)}(Q_2(S_2^\ast;\gpu^\ast))
\}
$.
If $C_2=0$, then $\gpu^\ast=\gpmin$.

We write $U_i(S_i)$ for the aggregate expected utility of spam transactions in sub-block $i$.
By symmetry among spam transactions in the same sub-block, setting this aggregate utility to zero is equivalent to setting the expected utility of each spam transaction in that sub-block to zero when $S_i>0$.
For the first sub-block, each opportunity created by user gas in that sub-block is captured by spam with probability $S_1/(S_1+1)$, by the same relative-ordering argument as in \cref{subsec:equilibrium-random}.
By linearity of expectation over all opportunities created in the first sub-block, the expected utility of spam in the first sub-block is
$
U_1(S_1;\gpu)
=
\frac{r_0}{D_0}Q_1(S_1;\gpu)\frac{S_1}{S_1+1}
-
S_1\spc\,\gp_1(S_1;\gpu).
$
The equilibrium spam volume in the first sub-block is the solution $S_1^\ast$ to
$
U_1(S_1^\ast;\gpu^\ast)=0,
$
subject to $S_1^\ast\spc\le C_1$.

For the second sub-block, spam revenue has two components.
First, opportunities may be created by user gas in the second sub-block and captured there, contributing
$
\frac{r_0}{D_0}Q_2(S_2;\gpu)\frac{S_2}{S_2+1}.
$
Second, opportunities may be created in the first sub-block but fail to be captured there.
The aggregate value of such surviving first-sub-block opportunities is
$
\frac{r_0}{D_0}\frac{Q_1^\ast}{S_1^\ast+1}.
$
These opportunities then reach the second sub-block and are captured there by the first spam transaction in that sub-block, whenever $S_2>0$.
Thus, the expected utility of spam in the second sub-block is
$
U_2(S_2;\gpu)
=
\frac{r_0}{D_0}
(
Q_2(S_2;\gpu)\frac{S_2}{S_2+1}
+
\frac{Q_1^\ast}{S_1^\ast+1}
)
-
S_2\spc\,\gpu.
$
The equilibrium spam volume $S_2^\ast$ is derived by solving
$
U_2(S_2^\ast;\gpu^\ast)=0,
$
subject to the capacity constraint $S_2^\ast\spc\le C_2$.

Together, the two zero-profit conditions and the fixed-point condition for $\gpu^\ast$ determine the equilibrium tuple $(S_1^\ast,S_2^\ast,\gpu^\ast)$.
The total spam volume is accordingly
$
S^\ast=S_1^\ast+S_2^\ast.
$
The total included user gas is
$
Q_u^\ast=Q_1^\ast+Q_2^\ast.
$

Figure~\ref{fig:pfo-spam-volume} illustrates how the resulting spam volume changes with $\bsmax$ and with the capacity-split parameter $v$.
For large $\bsmax$, once spam reaches its plateau, the $v=0.5$ curve has less spam than the random-ordering case ($v = 1$).
The reason is that spam that is included in the first sub-block pays the higher price $\gp_1$, while the cheaper lower region has only half of the block capacity.
For smaller $\bsmax$, however, $v=0.5$ can generate more spam than $v=1$, since splitting the block creates a lower-priced second sub-block even when the overall block is still relatively scarce.
Spam can enter this cheaper tail region, pay only $\gpu$, and still capture opportunities that survive from the first sub-block.
We note that this is an artifact of our model.

\parhead{Takeaway (Priority fee ordering)}
Priority fee ordering reduces spam when enough scarce capacity is allocated to a high-value early block region.
When the early region is small or has little demand, the lower sub-block can remain cheap, giving spammers an attractive tail region.

\begin{figure}[h]
    \centering
    \includegraphics{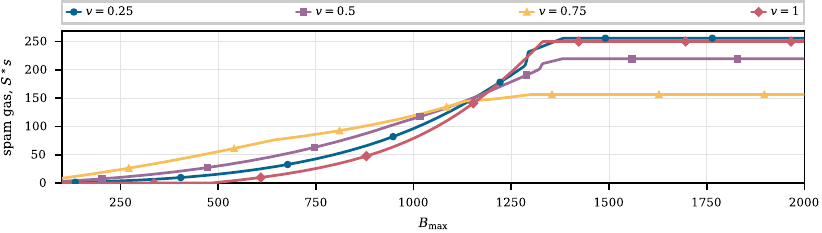}
    \caption{Equilibrium spam volume under the two-sub-block approximation to PFO for $v\in\{0.25,0.5,0.75,1\}$.
    We use the same calibration as in \cref{fig:spam-volume}.
    The case $v=1$ collapses to the random-ordering benchmark.}
    \Description{Four curves plot equilibrium spam gas against maximum block size for different allocations to the first sub-block. Each curve eventually reaches a plateau, but the curves cross as capacity increases. The allocation of 0.75 produces the lowest large-block plateau, followed by 0.5; the 0.25 and random-ordering curves have similar, higher plateaus.}
    \label{fig:pfo-spam-volume}
\end{figure}

\subsection{Metrics under Approximate PFO}
\label{subsec:pfo-welfare}

The discussion above focuses on spam volume and where spam is placed inside the block.
For completeness, we also evaluate user welfare, validator revenue, and network externality under the two-sub-block PFO model.
In the interest of space, we defer the formal expressions and numerical plots to \cref{app:pfo-metrics}.

The main qualitative takeaway is that for any fixed $v$, the pattern in the metrics mirrors the random-ordering model.
As $\bsmax$ rises, user welfare rises and eventually plateaus, validator revenue first increases but then decreases and finally plateaus, and network externality increases.
We avoid reading the separate $W_{\text{user}}$, $R$, and $E$ curves as direct comparisons across different values of $v$, since different $v$ changes the distribution of user valuations across sub-blocks.
The combined quantity $W_{\mathrm{user}}+R$ changes less as $v$ differs.

\section{Impact of Spam on Blockchain Scaling}
\label{sec:demand-scaling}
One of the central concerns with spam (e.g., the one of Flashbots~\cite{miller-mev-scaling-2025}) is that it may \emph{limit scaling}, as additional block capacity could simply be consumed by spammers instead of benefiting users.
Namely, the concern is that the chain may provision more block space, but the effective throughput available to genuine users may not increase proportionally if spam absorbs the added capacity.
Recall that spam share (\cref{subsec:metrics}) is a metric measuring the fraction of included gas consumed by spam MEV.
We now study whether and to what extent spam may limit the benefit of scaling to be passed down to a rise in user demand.

According to empirical data (see \cref{app:empirical-opportunity}), the total size of MEV profit grows with non-MEV trading volume, approximately linearly on Base and super-linearly on Arbitrum.
In our model, this corresponds to the aggregate value of all non-interacting opportunities created by genuine user activity.
As shown in \cref{subsec:equilibrium-random}, under random ordering the expected revenue of spam depends on the opportunity set only through this aggregate value, by linearity of expectation.
Accordingly, in this section, we assume that the aggregate opportunity value scales linearly with included user gas.

To formalize the problem further, let $\lambda\ge 1$ denote a demand-scaling parameter.
We scale non-spam demand as
$D_\lambda(\gp)=\lambda(D_0-\beta \gp)$,
so that the number of users at every price level grows proportionally with $\lambda$.
The realized aggregate opportunity value is assumed to be
$
\bar r_\lambda
:=
r_0\cdot \frac{Q_u}{D_0},
$
where $Q_u$ is the equilibrium amount of included user gas under the scaled demand curve.
Thus, as more user gas is included, the aggregate value of MEV opportunities grows proportionally.
Additionally, let $\rho_{\mathrm{spam}}$ denote the {\em spam share} of included gas, i.e., spam gas divided by total included gas.
A larger $\rho_{\mathrm{spam}}$ means that a greater fraction of the scaled block is unavailable to genuine users.

Here, we focus on the random-ordering benchmark and the two block-size rules from \cref{subsec:parameter-guidance}; the demand-scaling analysis under approximate PFO is deferred to \cref{app:demand-scaling-pfo}.
Under the baseline parameter choice, the designer sets $\bsmax$ equal to the plateau threshold under the scaled demand curve.
Since all users willing to pay the floor are then included, the aggregate opportunity value at the plateau is
$
\bar r_\lambda^\dagger=r_0\cdot D_\lambda(\gpmin)/D_0,
$
and the corresponding plateau block size is
$
\bs_{\mathrm{plat},\lambda}
=
D_\lambda(\gpmin)
+
\left(
\frac{r_0 D_\lambda(\gpmin)}{D_0\gpmin}
-
\spc
\right)_+.
$
Let $\spm_\lambda^\dagger$ denote the equilibrium spam volume at $\bsmax=\bs_{\mathrm{plat},\lambda}$.
The spam share of included gas under the baseline rule is then
$
\rho_{\mathrm{spam}}^\dagger(\lambda)
=
\frac{\spc\,\spm_\lambda^\dagger}{\bs_{\mathrm{plat},\lambda}}.
$

We also study how the refined block-size rule in \cref{subsec:parameter-guidance} behaves under the same scaling model.
Fix a floor $\gpmin$ and a target $\eta\in(0,1]$ for the minimum marginal user share.
Let $\bsmax^\ast(\gpmin,\eta)$ denote the resulting MMUS block size under the scaled demand curve, and let $\spm_\lambda^\ast$ be the corresponding equilibrium spam volume.
The spam share of the included gas is then
$
\rho_{\mathrm{spam}}^\ast(\lambda)
=
\frac{\spc\,\spm_\lambda^\ast}{\bsmax^\ast(\gpmin,\eta)},
$
and the total included gas equals $\bsmax^\ast(\gpmin,\eta)$.

Figure~\ref{fig:demand-scaling} plots $\rho_{\mathrm{spam}}^\dagger(\lambda)$ and $\rho_{\mathrm{spam}}^\ast(\lambda)$ for $\lambda\in[1,50]$, with $\eta=0.6$ and a fixed price floor $\gpmin=20$.
The main pattern is that, under linear aggregate-opportunity scaling, the spam share rises at first and then plateaus at a positive level.
The intuition is simple: both user demand at the floor and the aggregate value of induced MEV opportunities grow linearly with $\lambda$, so the amount of profitable spam also scales linearly.
As a result, scaling does not drive the spam share to zero.
Instead, the system approaches a regime in which spam remains a stable fraction of included gas.

We also observe that, for the same $\lambda$, the refined block-size choice $\bsmax^\ast$ yields a lower spam share than the baseline choice $\bs_{\mathrm{plat},\lambda}$.
The reason is that the MMUS rule stops scaling earlier, namely at the point where the marginal benefit of extra capacity becomes too spam-heavy.
Therefore, the refined rule keeps the spam share systematically below the $\bs_{\mathrm{plat}}$ benchmark.

\begin{figure}[h]
    \centering
    \includegraphics{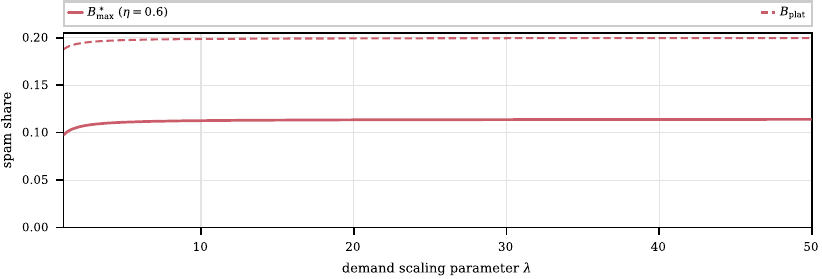}
    \caption{Impact of demand scaling on the spam share of included gas under random ordering.
    The x-axis and y-axis are the demand-scaling parameter $\lambda$ and the fraction of included gas consumed by spam, respectively.
    The dashed curve shows the outcome under the baseline rule that scales block size to $\bs_{\mathrm{plat},\lambda}$.
    The solid curve shows the outcome under the refined MMUS rule.
    We fix $\gpmin=20$ and $\eta=0.6$.}
    \Description{A line chart plots spam's share of included gas against the demand-scaling parameter. The dashed baseline curve approaches a plateau near one fifth. The solid refined-MMUS curve stays lower and approaches a plateau slightly above one tenth. Both shares remain positive as demand grows.}
    \label{fig:demand-scaling}
\end{figure}

\parhead{Takeaway (Spam under demand scaling)}
When aggregate MEV opportunity value scales linearly with included user gas, spam's share of included gas does not vanish as demand grows.
Instead, under both the plateau rule and the refined MMUS rule, the spam share approaches a sizable positive level.

\section{Empirical Analysis}\label{sec:empirical}

Our previous analysis predicts that spam absorbs a disproportionate share of marginal block capacity ($\bsmax$) and that the gas price floor ($\gpmin$) directly gates spam profitability.
We now examine how these two design levers affect spam in practice by analyzing daily spam gas on two major Ethereum Layer~2 rollups, Base and Arbitrum, which have undergone significant changes to both parameters during our observation period.
Following prior work~\cite{miller-mev-scaling-2025,solmaz2025optimistic}, we classify a contract as spam in a given month if it repeatedly probes for MEV opportunities but rarely trades. In practice, these are contracts whose transactions mostly call DEX pools or read price feeds without emitting any ERC-20 transfer. Spam gas is the gas of all transactions sent to such contracts. \Cref{app:data} states the rule in detail and audits what the largest such contracts do.
Because we classify each contract separately for each month and then count all transactions sent to a flagged contract during that month as spam, our measurements are intended to estimate aggregate spam volume and trends, rather than to classify every individual transaction with certainty.
We study Base and Arbitrum since both exhibit substantial spam MEV~\cite{solmaz2025optimistic} and both changed the parameters that our model studies, underwent relevant fee or block size changes, and provide publicly accessible historical on-chain data.

The Dune SQL queries used to collect the data are publicly accessible on \href{https://doi.org/10.5281/zenodo.23445675}{this page}.
\subsection{Spam on Base and Arbitrum}
\label{subsec:empirical-overview}

\Cref{fig:base-spam-full} shows spam's share of total gas on Base alongside the gas target and minimum gas price. Three phases are visible. First, after the Dencun upgrade in Mar~2024, Base progressively raised its gas target and spam's share rose sharply, peaking above 50\% in early 2025 (cf.\ \Cref{fig:base-spam-growth} for an indexed view of absolute gas). Second, when Base reduced the gas target from 70M to 50M in Jun~2025, spam's share declined (cf.\ \Cref{fig:base-target-decrease}). Third, starting in Dec~2025, Base introduced and progressively raised a protocol minimum gas price; spam's share fell sharply, to mostly below 10\% from Mar~2026 through Jun~2026, after a temporary rebound to about 21\% in early Feb~2026 (cf.\ \Cref{fig:base-spam-decline}).

\Cref{fig:arb-spam-full} shows the corresponding spam share for Arbitrum. Both rollups charge an L2 execution fee plus an L1 data posting surcharge; before Dencun, the L1 component dominated, but blobs made it negligible. The gas price floors discussed here differ in scope: Arbitrum's floor of 0.01~gwei (set at Dencun, down from an effective minimum of approximately 0.1~gwei when L1 costs dominated) covers the total fee, i.e., both execution and data posting, whereas Base's floor (introduced in Dec~2025) applies only to the L2 execution fee. With its higher floor, Arbitrum experienced lower spam than Base until Aug~2025 (mostly below 25\% after Dencun). On most days from Sep~2025 to Apr~2026, Arbitrum's share exceeded Base's. In Jan~2026, Arbitrum doubled the L2 gas price floor to 0.02~gwei as part of a broader fee mechanism overhaul~\cite{arb-fee-tweet-2025, arbitrum-dia-2026}. As noted in \Cref{sec:ecosystem-response}, spam's share dipped after the change; a brief rebound in mid-February proved transient (cf.\ \Cref{tab:arb-fee-floor}), and spam subsequently settled below pre-change levels.

\begin{figure*}[t]
    \centering
    \begin{subfigure}[t]{0.49\linewidth}
        \centering
        \includegraphics{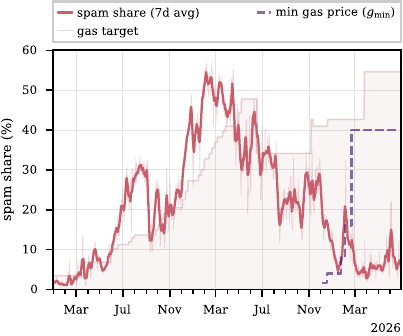}
        \caption{Spam's share of total gas on Base. The background shows the gas target and the stepwise protocol minimum gas price. Spam share peaks above 50\% in early 2025, then falls sharply after the gas price floor is introduced in Dec~2025.}
        \label{fig:base-spam-full}
    \end{subfigure}
    \hfill
    \begin{subfigure}[t]{0.49\linewidth}
        \centering
        \includegraphics{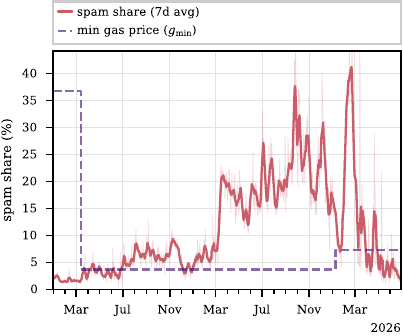}
        \caption{Spam's share of total gas on Arbitrum. The background shows the protocol minimum gas price, doubled from 0.01 to 0.02~gwei on Jan~9, 2026. The mid-February rebound is a transient episode driven by three short-lived bot contracts (cf.\ \Cref{tab:arb-fee-floor}).}
        \label{fig:arb-spam-full}
    \end{subfigure}
    \caption{Spam's share of total gas on Base and Arbitrum (7-day moving average, Jan~2024 to Jun~2026).}
    \Description{Two time-series panels plot spam's percentage of total gas on Base and Arbitrum. Base's share rises strongly through 2024, peaks in early 2025, and declines into 2026, with brief rebounds. Arbitrum's share is lower in 2024, rises during 2025, and falls in 2026 after a temporary February spike. Background step lines show Base's gas target and each chain's minimum gas price; thick spam curves show moving averages over thin daily traces.}
    \label{fig:spam-full}\vspace{-10pt}
\end{figure*}

\begin{remark}[Cross-chain comparisons and ordering.]
A raw comparison between Base and Arbitrum should not be read as a test of priority fee ordering.
Our approximate PFO analysis is comparative-static: it changes the ordering rule while holding fixed block capacity, gas price floors, user demand, MEV opportunities, and the distribution of demand across block positions.
Base and Arbitrum differ along all of these dimensions, so differences in their spam shares do not isolate the effect of ordering alone.
We therefore use the empirical evidence mainly in two ways: within-chain changes identify how spam responds to capacity and fee floors, while Base's within-block spam location provides qualitative evidence consistent with the PFO mechanism.
\end{remark}

\parhead{Spam location within blocks on Base.}
We also study where spam transactions appear within Base blocks from Jan~2024 to Jun~2025, before Base enabled Flashblocks on Jul~7, 2025, which split each block into 200\,ms sub-blocks that are ordered and streamed separately.
We collect all transactions sent to the identified spam contracts in each month, rank them by transaction index within the block after excluding the system transaction, and for each $k\in\{0,1,\dots,100\}$ compute the cumulative share of spam gas contained in the first $k\%$ of block positions.
Then we aggregate the monthly curves using spam gas as weights and report the interquartile range (IQR).

\Cref{fig:base-spam-block-position} shows that the cumulative spam gas curve lies below the uniform benchmark for most of the block. The first half of the block contains 42\% of spam gas, and the curve lies below the benchmark at the block midpoint in 15 of 18 months.
Spam is underrepresented in early block positions and concentrated in later positions.
This qualitative pattern is consistent with the analysis in \cref{sec:priority-fee-ordering}, as spam tends to appear in later, cheaper parts of the block rather than being spread uniformly.

\begin{figure}[h]
    \centering
    \includegraphics{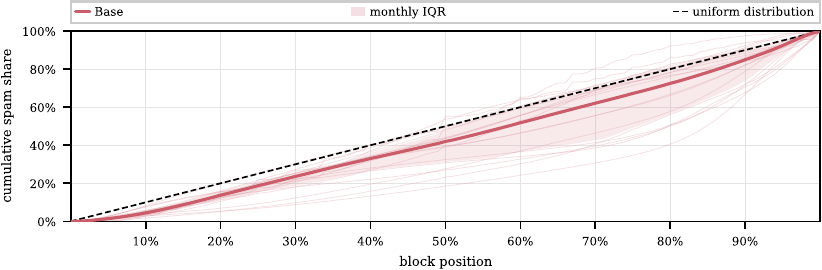}\vspace{-6pt}
    \caption{Distribution of spam gas by normalized position within Base blocks, Jan~2024 to Jun~2025.
    The x-axis is block position and the y-axis is the cumulative share of spam gas contained at or below that position.
    The solid line is the gas-weighted aggregate across months, thin lines are individual months, and the shaded band shows the monthly interquartile range (IQR), i.e., the band between the 25th and 75th percentiles of the monthly curves.
    The dashed 45-degree line is the uniform benchmark.
    A curve below the dashed line indicates that spam is concentrated toward later positions in the block.}
    \Description{Cumulative spam gas share is plotted against normalized transaction position within a block. The thick aggregate curve lies below the dashed diagonal uniform benchmark through most positions, indicating concentration toward the block's end. Thin monthly curves and a shaded interquartile band show variation across months. All curves reach the full cumulative share at the end of the block.}
    \label{fig:base-spam-block-position}
\end{figure}

\subsection{Spam and Block Capacity}
\label{subsec:empirical-capacity}

We regress spam gas on the gas target, i.e., the empirical counterpart of $\bsmax$, using the full sample of 912 daily observations from Base. We exclude the market-determined clearing fee because it is endogenous: for most of the sample period there was no protocol gas price floor, so the observed minimum fee reflects demand-driven congestion rather than an exogenous policy variable ($\gpmin$ in our model). \Cref{tab:spam-regression} reports the results.

Model~1 shows that spam gas scales super-linearly with the gas target: a 1\% increase in the gas target is associated with a 1.64\% increase in spam gas ($R^2 = 0.62$). Model~2 shows that higher gas targets also increase spam's share of total gas ($+0.07$ per log-unit, $R^2 = 0.15$), consistent with spam disproportionately filling marginal capacity (see \cref{fig:marginal-capacity} and associated analysis in \cref{subsec:parameter-guidance}). %
The gas target reduction on Jun~18, 2025 (70M to 50M) provides a concrete illustration: spam gas fell by 34\% compared to 24\% for non-spam gas in the 30-day window (cf.\ \Cref{fig:base-target-decrease}), i.e., both declined, but spam absorbed a larger share of the reduction.

Since daily observations exhibit serial dependence, we report Newey-West standard errors (6 lags) to avoid understating the true uncertainty. Both coefficients remain highly significant ($p < 0.001$) after this correction. These regressions capture only Base; Arbitrum did not change its gas target during the observation period.

\begin{table}[t]
    \centering
    \footnotesize
    \newlength{\rowpad}\setlength{\rowpad}{\dimexpr(\aboverulesep+\belowrulesep+\lightrulewidth)/4\relax}
    \begin{minipage}[t]{0.375\linewidth}
        \centering
        \captionsetup{justification=raggedright,singlelinecheck=false}
        \setlength{\tabcolsep}{2.5pt}
        \begin{tabular}[t]{lcc}
            \toprule
            & \shortstack{\textbf{(1)}\\$\log(\text{spam gas})$} & \shortstack{\textbf{(2)}\\spam share} \\
            \midrule
            $\log(\text{gas target})$ & $1.64^{***}$ & $0.07^{***}$ \\
                                       & $(0.131)$ & $(0.011)$ \\
            Constant                   & $-2.42$ & $-0.96^{***}$ \\
                                       & $(2.260)$ & $(0.185)$ \\
            \midrule
            $R^2$                      & $0.62$ & $0.15$ \\
            \bottomrule
        \end{tabular}
        \vspace{2pt}
        \caption{OLS regression results for spam on Base (912 daily observations, Jan~2024 to Jun~2026). Model~1 regresses $\log(\text{spam gas})$ on $\log(\text{gas target})$. Model~2 regresses spam's share of total gas on the same regressor. Newey-West standard errors (6 lags) in parentheses. $^{***}$ $p < 0.001$.}
        \label{tab:spam-regression}
    \end{minipage}\hfill
    \begin{minipage}[t]{0.6\linewidth}
        \centering
        \setlength{\tabcolsep}{3pt}
        \begin{tabular}[t]{lcccc}
            \toprule
            \shortstack{Floor step\\(gwei)} & \shortstack{Window\\(days)} & \shortstack{Spam\\$\Delta$} & \shortstack{Non-spam\\$\Delta$} & \shortstack{Spam\\share} \\
            \midrule
            $0 \to 0.0002$ (Dec~5)  & 13 & $-40.2\%$  & $+9.0\%$  & $26.9\% \to 16.8\%$ \\[\rowpad]
            $\to 0.0005$ (Dec~18)   & 36 & $-55.4\%$  & $+19.2\%$ & $22.5\% \to 9.8\%$ \\[\rowpad]
            $\to 0.001$ (Jan~23)    & 11 & $+103.6\%$ & $-6.3\%$  & $5.7\% \to 11.6\%$ \\[\rowpad]
            $\to 0.002$ (Feb~3)     & 17 & $+48.3\%$  & $-6.2\%$  & $9.7\% \to 14.5\%$ \\[\rowpad]
            $\to 0.005$ (Feb~20)    & 9  & $-19.7\%$  & $-2.7\%$  & $11.8\% \to 9.9\%$ \\
            \bottomrule
        \end{tabular}
        \vspace{2pt}
        \caption{Non-overlapping before-and-after comparisons around each gas price floor step on Base. The ``after'' window runs from the step to the next step (for the last step, to Feb~28, 2026); the ``before'' window is an equal-length period immediately preceding the step. $\Delta$ is the percentage change in daily averages.}
        \label{tab:fee-floor-snapshots}
    \end{minipage}\vspace{-16pt}
\end{table}

\subsection{Spam and The Minimum Gas Price}
\label{subsec:empirical-fee-floor}

The capacity regressions above deliberately exclude the minimum fee because the observed fee is endogenous for most of the sample. To isolate the effect of the protocol gas price floor ($\gpmin$ in our model), we focus on periods where an exogenous floor was introduced or raised and the gas target was largely constant.

\parhead{Base.} Base introduced a protocol minimum gas price in Dec~2025 and raised it in a series of steps (0.0002, 0.0005, 0.001, 0.002, and 0.005~gwei). Since the steps are unevenly spaced (11 to 36 days apart), fixed-length before/after windows would overlap. We therefore use non-overlapping windows: the ``after'' period runs from one step to the next, and the ``before'' period is a window of equal length immediately preceding the step. Results are reported in~\cref{tab:fee-floor-snapshots}.

Over the full gas price floor period, spam's share of total gas fell from 27\% to around 10\%, though the decline was not monotonic across individual steps. The introduction of the gas price floor and the first increase to 0.0005~gwei coincided with the largest reductions ($-40\%$ and $-55\%$), while non-spam gas grew in both cases ($+9\%$ and $+19\%$).
The 0.001 and 0.002~gwei steps saw spam rebound, possibly due to confounding factors such as increased price volatility creating more arbitrage opportunities during that period. The subsequent increase to 0.005~gwei again coincided with a 20\% spam reduction with negligible non-spam impact ($-3\%$). The data is limited to five steps over a short period, but the overall trend is clear: spam's share declined substantially, and the individual steps where spam fell did so while non-spam gas was unaffected or grew, consistent with the gas price floor raising costs for spam bots while remaining negligible for organic users.

\parhead{Arbitrum.} Arbitrum doubled its gas price floor from 0.01 to 0.02~gwei on Jan~9, 2026. \Cref{tab:arb-fee-floor} compares symmetric 90-day windows before and after the change.

\setlength{\intextsep}{4pt}
\begin{wraptable}{r}{0.6\linewidth}
    \centering
    \small
    \setlength{\tabcolsep}{4pt}
    \resizebox{\linewidth}{!}{%
    \begin{tabular}{lccc}
        \toprule
        & Before (90d) & After (90d) & $\Delta$ \\
        \midrule
        Spam gas (Bgas/day)     & $101.8$ & $78.5$  & $-22.9\%$ \\
        Non-spam gas (Bgas/day) & $372.3$ & $313.1$ & $-15.9\%$ \\
        Spam share              & $21.5\%$ & $20.0\%$ & $-6.7\%$ \\
        \bottomrule
    \end{tabular}}
    \vspace{2pt}
    \caption{90-day before-after comparison around the Arbitrum gas price floor doubling (0.01 to 0.02~gwei) on Jan~9, 2026. Before: Oct~11 to Jan~8; after: Jan~9 to Apr~8.}
    \label{tab:arb-fee-floor}\vspace{-6pt}
\end{wraptable}

Over the 90-day windows, doubling the gas price floor reduced spam. Spam gas fell by 23\%, and its share of total gas declined from 22\% to 20\%. Non-spam gas also fell over the same period ($-16\%$). But it fell by less than spam, so spam absorbed a larger share of the overall decline. This is consistent with the gas price floor gating spam more than organic demand. The reduction was not immediate. In the first weeks after the change, spam briefly dipped and then rebounded above pre-change levels (cf.\ \Cref{sec:ecosystem-response}). Taken in isolation, this could suggest the floor had no lasting effect. The rebound, however, was transient. It began on Jan~29 and was sustained in February by three short-lived contracts created on Feb~7 and~8, which together accounted for 45\% of all Arbitrum spam gas in February.\footnote{The three contracts are \texttt{0x7c99\ldots 28ac}, \texttt{0xb4a1\ldots bf47}, and \texttt{0xc487\ldots 24d9}.} Two of them were deployed by the same address with identical bytecode and both stopped on Feb~25, indicating a single operator. The third was deployed by a different address with different code and continued at low volume until late March. Once this episode passes, spam settles well below pre-change levels. Over the longer window, the floor increase is therefore associated with a clear reduction in spam.

\parhead{Scope of the empirical analysis.}
Our empirical results are case studies of Base and Arbitrum and may not hold for all blockchain architectures.
Extending the same measurements to other chains may be limited by lower activity volumes of the relevant spam transactions on some systems and by the availability of comparable historical data on other active systems, such as Solana, at the time of our analysis.

\section{Mitigations to Spam}
\label{sec:mitigations}

In this section, we discuss mitigation directions suggested by our analysis.
We organize mitigations into two broad classes.
The first class consists of {system-level mechanisms}, which operate at the blockchain or block-producer layer.
Within this class, {\em incentive-based mechanisms} make spam more expensive, while {\em cheap-path mechanisms} make failed spam cheaper to process.
The second class consists of {\em application-specific mechanisms}, where specific applications are redesigned to reduce the MEV opportunities for spam.
These approaches can be combined in practice.

\subsection{Incentive-Based Mechanisms}
\label{subsec:incentive}

Incentive-based mechanisms reduce spam by increasing its expected cost, so that fewer spam attempts enter in equilibrium.
They do not require perfectly identifying spam and can be deployed as coarse policies targeting resource consumption.

\parhead{Minimum gas prices.}
A simple mechanism is to set a nontrivial floor gas price (i.e., increasing $\gpmin$).
As discussed in \cref{sec:ecosystem-response}, several chains have adopted nontrivial fee floors, either from launch or in response to spam~\cite{base-fee-increase-2025, arb-fee-tweet-2025, arbitrum-dia-2026, aptos-gas-fee-2025, sui-gas-2025}.
This is consistent with our analysis in \cref{sec:random-ordering}: a sufficiently high floor gas price caps spam volume and preserves useful block space for users. {Empirical evidence for this effect is shown in \cref{subsec:empirical-fee-floor}.}

\parhead{Charging for reserved execution gas.}
Another lever is to charge based on \emph{gas limit} rather than ex post gas used.
This targets strategies that reserve capacity but revert early or do little work when they fail.
In our model, this makes spam pay for the capacity it forces the network to reserve, which reduces equilibrium spam volume; Monad, for instance, adopted this approach~\cite{categorylabs-monad-spec-2025}.

\subsection{Cheap-Path Mechanisms}
\label{subsec:cheap-path}

Cheap-path mechanisms reduce the cost to the system of failed spam transactions.

\parhead{Filtering when execution is not the bottleneck.}
When the block producer can simulate transactions, it can filter out transactions that produce no state changes before including them in a block.
This is aligned with Ethereum's block-building pipeline, where builders and relays already simulate candidate blocks~\cite{flashbots-relay-docs}, and with future proposer-builder separation designs~\cite{eip-7732}.
A practical concern is that block producers themselves can become targets for spam, so rate limits, reputation systems, or out-of-band fees may still be needed~\cite{flashbots-relay-docs}.

\parhead{Cancellation when execution is the bottleneck.}
When execution resources are scarce, one possible approach is to allow a failed spam transaction to be {\em cancelled} without full execution, provided the cancellation is justified in a verifiable way.
For example, a sender could provide evidence that the transaction would lead to no state changes, while paying a lower but nonzero fee.

\subsection{Application-Specific Mechanisms}
\label{subsec:application-specific}

Application-specific mechanisms reduce spam by changing the protocol that generates the MEV opportunity, so that the opportunity size is reduced.
The application can internalize, auction, or redistribute the value that would otherwise be left to spam transactions.
Examples include Chainlink's Smart Value Recapture feeds for oracle-update backrunning~\cite{chainlink-svr-docs-2026}, and protocol redesigns that redistribute MEV at the application layer~\cite{zhang2025rediswap,adams2025amm}.

\section{Related Work}
\label{sec:related}
\parhead{Empirics of spam MEV.}
Spam MEV was first documented on Solana by Umbra Research~\cite{umbra2022solana}, who observed high failure rates among MEV transactions and provided early evidence that high-throughput, low-fee chains incentivize redundant submissions over priority-fee competition.
Solmaz et al.~\cite{solmaz2025optimistic} study optimistic cyclic arbitrage patterns on Ethereum Layer~2 rollups and find that cyclic arbitrage bots consumed over 50\% of gas on Base and Optimism by early 2025, with only 6--12\% of probes resulting in a trade.
Gogol et al.~\cite{gogol2025first} study reverted transactions on five rollups, finding that revert rates rose sharply after Dencun; by restricting to reverts, they capture one subtype of spam implementation and strategy, as not all unsuccessful probes revert.
Wu et al.~\cite{wu2026wait} study the distribution of spam transactions of targeted search, where bots identify opportunities off-chain and submit route-committed transactions, and of probabilistic search, where bots repeatedly probe on-chain and resolve opportunity discovery during execution.
These works are primarily empirical. Our work develops a framework that characterizes how block capacity, gas price floors, and the TFM jointly determine spam volume and analyzes the resulting impact on user welfare, validator revenue, and network externality. We complement the analysis with empirical evidence.

\parhead{Modeling spam MEV.}
Concurrent with our work, Mazorra et al.~\cite{mazorra2026timinggames} model spam as a timing game among a finite number of searchers competing for an opportunity, showing that timing competition fully dissipates expected profits. Their equilibrium spam volume in the large-$n$ limit corresponds exactly to ours, despite the models being derived independently. The two models are complementary: Mazorra et al.\ characterize when and how often each searcher submits in the timing game, while we focus on how protocol design parameters, i.e., block capacity, gas price floors, and transaction ordering, jointly shape spam volume and its welfare and externality implications.

{In related work, Capponi and Zhu~\cite{capponi2025auctioningtime} model costly duplicate submissions as a latency race among traders competing for time-sensitive opportunities on blockchains, framing the problem as an observable analogue of HFT latency investment. They show that a time-priority auction (Timeboost) reduces redundant submissions and reallocates the waste into platform revenue, and validate this prediction on Arbitrum using a difference-in-differences design. Their focus is on a specific auction mechanism as mitigation, whereas we study how general design parameters jointly shape spam volume and welfare.}

\parhead{Targeted MEV.}
Eskandari et al.~\cite{eskandari2019sok} taxonomized front-running attacks on Ethereum, and Daian et al.~\cite{daian2020flash} broadened the scope to MEV more generally, documenting priority gas auctions among searchers. Qin et al.~\cite{qin2021quantifying} and Torres et al.~\cite{ferreira2021frontrunner} subsequently quantified sandwich attacks, arbitrage, and liquidations at scale. A large body of further work has measured specific strategies and infrastructure on the Ethereum Layer~1~\cite{zhou2021high,zhou2021just,heimbach2024nonatomic,wang2022cyclic,weintraub2022flash,piet2022extracting,adams2024mev,zhu2024revert,wahrstatter2023time,heimbach2023pbs}, on alternative Layer~1s~\cite{oz2024fcfs}, and across chains and rollups~\cite{torres2024rolling,gogol2025nonatomic,oz2025pandora,obadia2021crossdomain}. These works study targeted MEV, i.e., strategies where searchers identify a specific opportunity off-chain and submit a transaction to capture it. Our work, in contrast, focuses on spam MEV, where searchers submit speculative transactions whose profitability is resolved only at execution time.

\parhead{MEV prevention.}
Proposed mitigations for targeted MEV include protocol-specific parameter tuning~\cite{heimbach2022eliminating,zhou2021a2mm}, trusted third-party ordering via relays or private mempools~\cite{flashbots,mevboost,protectrpc,cowswap}, fair ordering protocols~\cite{kelkar2020order,kelkar2023themis,kursawe2020wendy,cachin2021quick}, commit--reveal and privacy-preserving schemes~\cite{zhang2022f3b,momeni2023fairblock,breidenbach2018enter}, trusted execution environments~\cite{bentov2019tesseract,stathakopoulou2021adding}, and encrypted mempools via threshold encryption~\cite{choudhuri2024bte,choudhuri2025onetime,agarwal2025weightedbte,agarwal2025efficiently,bormet2025beat,babel2024prof}; see Heimbach and Wattenhofer~\cite{heimbach2022sok} for a comprehensive systematization. Spam MEV, however, is largely unaddressed by these mechanisms and may even be worsened: encrypted mempools, for instance, remove the information advantage that enables targeted extraction, potentially shifting searchers toward speculative, high-volume strategies instead. As part of this work, we also propose mitigations specifically targeting spam MEV.

\section{Conclusions and Future Directions}
We develop a framework for analyzing spam MEV under a competitive equilibrium, deriving equilibrium spam volumes and characterizing the impact on user welfare, validator revenue, and network externality. Empirical evidence from Base and Arbitrum is consistent with several predictions of the framework.

An important direction for future work is the concrete design of system-level mitigation mechanisms, both making spam less attractive to send and failed probes cheaper to process. How these approaches interact across blockchain architectures remains open.

\bibliography{main}
\bibliographystyle{ACM-Reference-Format}

\appendix

\section{Metrics for Random Ordering Model}
\label{app:random-ordering-metric}

This appendix records the level quantities and counterfactual gaps used in \cref{subsec:welfare-analysis}.
Throughout, we compare two worlds evaluated at the same design parameters $(\bsmax,\gpmin)$: the equilibrium world with spam, and the counterfactual world without spam.
This isolates the effect of spam while holding the block size and gas price floor fixed.

\parhead{Spam world and spam-free world.}
In the spam-free world, all block space is available to genuine users.
The clearing price is
$$
\gp^0(\bsmax,\gpmin)
=
\max\{\gpmin,\frac{D_0-\bsmax}{\beta}\},
$$
and the amount of user gas included in the block is
$$
Q_u^0(\bsmax,\gpmin)
:=
\min\{\bsmax,D(\gpmin)\}.
$$

In the spam world, some block space may be taken up by spam.
Using the equilibrium spam volume $\spm^\ast$ from \cref{subsec:equilibrium-random}, the amount of included user gas is
$$
Q_u^\ast(\bsmax,\gpmin)
:=
\min\{\bsmax-\spc \spm^\ast(\bsmax,\gpmin),D(\gpmin)\}.
$$
Thus, users get only the block space left after spam.
We have $Q_u^\ast\le Q_u^0$, and the gap $Q_u^0-Q_u^\ast$ is the useful block space displaced by spam.

Let $\gp^\ast(\bsmax,\gpmin)$ denote the clearing price in the spam world, and let
$$
G^\ast(\bsmax,\gpmin)
:=
Q_u^\ast(\bsmax,\gpmin)+\spc\spm^\ast(\bsmax,\gpmin)
$$
denote the total gas sold in equilibrium.
Equivalently, in the congested region $G^\ast=\bsmax$, while in the slack region $G^\ast=D(\gpmin)+\spc\spm^\ast$.

\parhead{User welfare.}
Recall that user welfare is the sum of genuine users' utilities.
Under the linear demand model, it depends only on how much user gas is included.
Therefore, in the spam-free world,
$$
W_{\text{user}}^0(\bsmax,\gpmin)
=
\frac{(Q_u^0(\bsmax,\gpmin))^2}{2\beta},
$$
while in the spam world,
$$
W_{\text{user}}^\ast(\bsmax,\gpmin)
=
\frac{(Q_u^\ast(\bsmax,\gpmin))^2}{2\beta}.
$$
These level expressions show the main mechanism: under random ordering, spam lowers user welfare by reducing useful user gas.

\parhead{Validator revenue.}
Validator revenue is the amount of gas sold times the clearing price.
Without spam, revenue comes only from users:
$$
R^0(\bsmax,\gpmin)
=
\gp^0(\bsmax,\gpmin) Q_u^0(\bsmax,\gpmin).
$$
In the spam world, validator revenue is
$$
R^\ast(\bsmax,\gpmin)
=
\gp^\ast(\bsmax,\gpmin) G^\ast(\bsmax,\gpmin).
$$
Comparing $R^\ast$ and $R^0$ captures how spam can increase validator revenue even when it harms users.

\parhead{Externality.}
We note that when choosing the parameters, we cannot ignore the cost that the parameter choices have on the overall blockchain ecosystem.
Along the lines of Buterin~\cite{buterin2018resource}, we model the ecosystem cost of processing blocks as
$$
E(\bsmax,\gpmin)=c_1 \bsmax+c_2 G,
$$
where $c_1$ is the cost of provisioning block capacity and $c_2$ is the per-gas execution cost.
In the spam-free world,
$$
E^0(\bsmax,\gpmin)
=
c_1 \bsmax+c_2 Q_u^0(\bsmax,\gpmin),
$$
while in the spam world,
$$
E^\ast(\bsmax,\gpmin)
=
c_1 \bsmax+c_2 G^\ast(\bsmax,\gpmin).
$$
The capacity term is shared by both worlds, while the execution term can be larger in the spam world.
Note that this additional burden is not only borne by the validators, but by the overall ecosystem, such as higher hardware and bandwidth requirements for full nodes, which has second-order externalities on decentralization.
We use the term \emph{externality} following Buterin~\cite{buterin2018resource}, who draws a parallel to environmental pollution: although the block producer is compensated for including transactions, the cost of disseminating, executing, and storing them is borne by every full node in the network, much like pollution produced by one factory is suffered by the entire surrounding area.
The quantity $E(\bsmax,\gpmin)$ captures this collective burden.

\subsection{Counterfactual Delta}
We now record the gaps between the spam world and the spam-free world.
These deltas isolate the effect of spam itself while holding $(\bsmax,\gpmin)$ fixed.

\parhead{User welfare delta.}
The user-welfare change caused by spam is
\begin{align*}
\Delta W_{\text{user}}(\bsmax,\gpmin)
&=
W_{\text{user}}^\ast(\bsmax,\gpmin)
-
W_{\text{user}}^0(\bsmax,\gpmin)\\
&=
\frac{
(Q_u^\ast(\bsmax,\gpmin))^2
-
(Q_u^0(\bsmax,\gpmin))^2
}{2\beta}.
\end{align*}
This quantity is weakly negative.
It is zero when spam does not enter, and also when the block is large enough that spam no longer displaces users.

\parhead{Validator revenue delta.}
The change in validator revenue caused by spam is
$$
\Delta R_{\text{val}}(\bsmax,\gpmin)
=
R^\ast(\bsmax,\gpmin)-R^0(\bsmax,\gpmin).
$$
In our model, this term is always weakly positive: spam can raise validator revenue by increasing the clearing price, by filling otherwise empty block space, or by doing both.

\parhead{Externality delta.}
The increase in ecosystem cost caused by spam is
\begin{align*}
\Delta E(\bsmax,\gpmin)
&=
E^\ast(\bsmax,\gpmin)-E^0(\bsmax,\gpmin)\\
&=
c_2\bigl(G^\ast(\bsmax,\gpmin)-Q_u^0(\bsmax,\gpmin)\bigr).
\end{align*}
The capacity term $c_1\bsmax$ cancels because both worlds use the same block size.
Thus, $\Delta E$ measures the additional execution burden created by spam.

Figure~\ref{fig:welfare-deltas} summarizes these counterfactual gaps.

\begin{figure*}[h]
    \centering
    \includegraphics{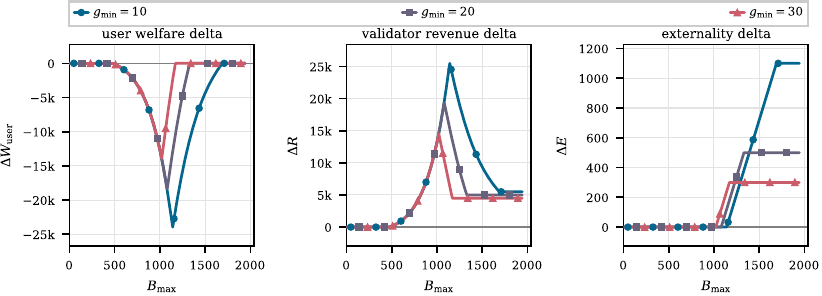}
    \caption{Counterfactual impact of spam as a function of block size. Negative $\Delta W_{\text{user}}$ means that spam lowers user welfare. Positive $\Delta R_{\text{val}}$ and $\Delta E$ mean that spam raises validator revenue and network externality, respectively.}
    \Description{Three panels plot the difference between spam and no-spam outcomes against block size, with separate curves for three minimum gas prices. The welfare difference falls below zero and then returns to zero. The revenue difference rises to a peak and then settles at a positive plateau. The externality difference remains zero initially, then rises to a positive plateau. The lowest minimum gas price produces the largest absolute differences.}
    \label{fig:welfare-deltas}
\end{figure*}

\parhead{User welfare loss peaks at block capacity of $D(\gpmin)$.}
The following result characterizes where spam is most harmful to users.

\begin{proposition}[Welfare Loss Peaks at $D(\gpmin)$]
\label{prop:welfare-peak}
Under the linear demand model $D(\gp) = D_0 - \beta\gp$, the user welfare loss $\Delta W_{\text{user}}(\bsmax)$ is most negative (i.e., users are most harmed) at $\bsmax = Q_{\min}$, where $Q_{\min}:=D(\gpmin)$.
\end{proposition}

The proof is in \cref{app:proofs}.
For $\bsmax > Q_{\min}$, the spam-free welfare is constant while the spam-world welfare increases, so the gap shrinks.
For $\bsmax < Q_{\min}$, every marginal unit of capacity serves a user in the spam-free world, but some is absorbed by spam, so the gap widens.

\section{Choosing $\gpmin$ for a fixed $\bsmax$}
\label{app:param-gmin}

The discussion above fixes $\gpmin$ and chooses $\bsmax$.
In practice, the opposite situation also occurs: the designer may want to keep $\bsmax$ fixed and choose the gas price floor instead.
This is natural when block size is costly to change, or when the system operator wants to adjust spam volume through the gas price floor rather than through throughput.

\parhead{A baseline rule.}
Fix a block size $\bsmax$.
The simplest rule is to choose the smallest $\gpmin$ such that the current block size already equals the plateau threshold, that is,
$$
\gpmin^\dagger(\bsmax)
\ \text{such that}\ 
B_{\mathrm{plat}}(\gpmin^\dagger)=\bsmax.
$$
Because $B_{\mathrm{plat}}(\gpmin)$ is strictly decreasing in $\gpmin$, this rule defines a unique gas price floor.
If the floor is already high enough that spam does not enter, then the baseline choice is simply
$$
\gpmin^\dagger(\bsmax)=\frac{D_0-\bsmax}{\beta}.
$$
When the solution lies in the slack-with-spam region, we have
$$
\gpmin^\dagger(\bsmax)
=
\frac{
-(\bsmax-D_0+\spc+\beta r_0 / D_0)
+
\sqrt{(\bsmax-D_0+\spc+\beta r_0 / D_0)^2+4\beta r_0}
}{2\beta}.
$$
If the gas price floor were lower than $\gpmin^\dagger(\bsmax)$, then the current block size would be too small relative to the amount of spam and non-spam demand that wants to enter.
If the gas price floor were higher than $\gpmin^\dagger(\bsmax)$, then the block would be slack, but the system would be using a stronger fee floor than needed to reach that point.

\parhead{Choosing $\gpmin$ with the refined rule.}
The baseline rule makes the current block just sufficient, but it does not control how the newly admitted block usage is split between users and spam as the gas price floor is lowered.
Similar to the previous analysis, we define the user share of the newly admitted used capacity when the gas price floor is reduced from $\gpmin$ to $\gpmin-\Delta g$ (assume $\Delta g\to 0$):
$$
\mu_{\text{user}}
:=
\frac{-\frac{\partial Q_u^\ast(\bsmax,\gpmin)}{\partial \gpmin}}
{-\frac{\partial Q_u^\ast(\bsmax,\gpmin)}{\partial \gpmin}
-\frac{\partial (\spc\,\spm^\ast(\bsmax,\gpmin))}{\partial \gpmin}}.
$$
The minus signs appear because we are studying the effect of lowering the gas price floor.

\begin{proposition}
\label{prop:mu-user-gmin}
Fix $\bsmax$.
Then the local user share from lowering $\gpmin$ has the following form with the linear demand function:
\begin{enumerate}[leftmargin=*]
    \item If
    $
    \gpmin \ge \max\{\frac{D_0-\bsmax}{\beta},\frac{r_0 D_0}{D_0\spc+\beta r_0}\},
    $
    then no spam enters and $\mu_{\text{user}}=1$.
    In this region, lowering the gas price floor only admits additional users.
    \item If
    $
    \gpmin^\dagger(\bsmax)\le \gpmin < \frac{r_0 D_0}{D_0\spc+\beta r_0},
    $
    then the block remains slack and spam enters.
    In this region, lowering $\gpmin$ strictly decreases the fraction of newly admitted used capacity that goes to users.
    \item If
    $
    \frac{D_0-\bsmax}{\beta}<\gpmin<\gpmin^\dagger(\bsmax),
    $
    then the spam world is congested.
    In this region, the denominator in the definition of $\mu_{\text{user}}$ is zero.
\end{enumerate}
\end{proposition}

The proof can be found in \cref{app:proofs}.
\Cref{prop:mu-user-gmin} gives a simple refined rule for choosing $\gpmin$.
Fix a target $\eta\in(0,1]$ and require that at least an $\eta$ fraction of newly admitted used capacity go to users.
Inside the slack-with-spam region, the condition $\mu_{\text{user}}\ge \eta$ becomes
$$
\frac{\beta\,\gpmin^2}{\beta\,\gpmin^2+r_0}\ge \eta,
$$
which is equivalent to
$$
\gpmin \ge \sqrt{\frac{\eta r_0}{\beta(1-\eta)}}.
$$
Therefore, if the block size $\bsmax$ is fixed, a refined gas price floor choice is
\[
\displaystyle
\gpmin^\star(\bsmax,\eta)
=
\max\{
\gpmin^\dagger(\bsmax),
\min\{\frac{r_0 D_0}{D_0\spc+\beta r_0},\sqrt{\frac{\eta r_0}{\beta(1-\eta)}}\}
\}.
\]
In this equation, the term $\gpmin^\dagger(\bsmax)$ is the baseline floor that makes the current block size just sufficient, and the second term is the user-share threshold.
Therefore, the designer lowers $\gpmin$ only until either the current block size would cease to be sufficient, or the newly added capacity would become too spam-heavy.

\section{Results of Exponential Demand}
\label{app:exp-demand}

In the main body, we use the linear demand function $D(\gp)=D_0-\beta\gp$ because it gives closed-form expressions for the equilibrium spam volume and the welfare metrics.
Here we record the corresponding analysis under an exponential demand curve.
The purpose of this appendix is to show that the qualitative behavior in \cref{sec:random-ordering} does not rely on linear demand.

Let
$
D^{\exp}(\gp)=D_0 e^{-\lambda \gp}
$
for some $\lambda>0$.
The inverse demand curve is
$
P^{\exp}(Q)=\frac{1}{\lambda}\log\frac{D_0}{Q}.
$
As in the main text, let $Q_u$ denote included genuine-user gas and let the aggregate value of all non-interacting MEV opportunities scale linearly with included user gas:
$
r=r_0\cdot Q_u/D_0.
$
The floor demand is
$
Q_{\min}^{\exp}:=D^{\exp}(\gpmin)=D_0 e^{-\lambda\gpmin}.
$

\begin{figure}[h]
    \centering
    \includegraphics{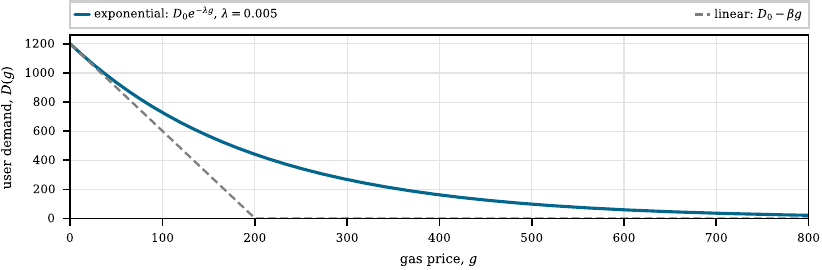}
    \caption{Exponential demand curve $D^{\exp}(\gp)=D_0e^{-\lambda \gp}$.
    }
    \Description{User demand is plotted against gas price. A solid exponential curve declines smoothly from its initial demand and approaches zero gradually. A dashed linear comparison starts at the same demand, reaches zero at a finite gas price, and remains zero thereafter. The exponential curve lies above the linear comparison for positive prices.}
    \label{fig:app-exp-demand}
\end{figure}

\parhead{Clearing price at a fixed spam volume.}
Fix a spam volume $\spm$.
Since each spam transaction reserves $\spc$ gas, the remaining capacity available to users is $\bsmax-\spc\spm$.
The amount of included user gas is
$$
Q_u(\spm)
=
\min\{
(\bsmax-\spc\spm)^+,
Q_{\min}^{\exp}
\}.
$$
The corresponding clearing price is
$$
\gp(\spm)
=
\max\{
\gpmin,
P^{\exp}((\bsmax-\spc\spm)^+)
\},
$$
where the inverse-demand term is used only when $\bsmax-\spc\spm>0$.
Equivalently, if $\bsmax-\spc\spm\ge Q_{\min}^{\exp}$, then the floor binds and $\gp(\spm)=\gpmin$.
If $0<\bsmax-\spc\spm<Q_{\min}^{\exp}$, then the block is congested and
$$
\gp(\spm)
=
\frac{1}{\lambda}
\log\frac{D_0}{\bsmax-\spc\spm}.
$$

Given this price, the expected utility of each spam transaction is
$$
u(\spm)
=
\frac{1}{\spm+1}\cdot
\frac{r_0 Q_u(\spm)}{D_0}
-
\spc \gp(\spm).
$$
The first term is the expected MEV revenue of each spam transaction, where the aggregate opportunity value is proportional to included user gas.
The second term is the inclusion cost.
As before, a competitive equilibrium is pinned down by the zero-profit condition $u(\spm^\ast)=0$, together with the no-entry condition when $u(0)\le 0$.

Let
$
Q_u^0=\min\{\bsmax,Q_{\min}^{\exp}\}$ and
$\hat{\gp}
=
\max\{
\gpmin,
P^{\exp}(Q_u^0)
\}
$
denote the spam-free user gas and the spam-free clearing price.
Spam does not enter if
$
r_0\frac{Q_u^0}{D_0}\le \spc \hat{\gp}.
$

Unlike the linear-demand case, this root generally does not have a simple closed form, so we solve it numerically, and show the results in \cref{fig:app-exp-spam-volume}.

\begin{figure}[h]
    \centering
    \includegraphics{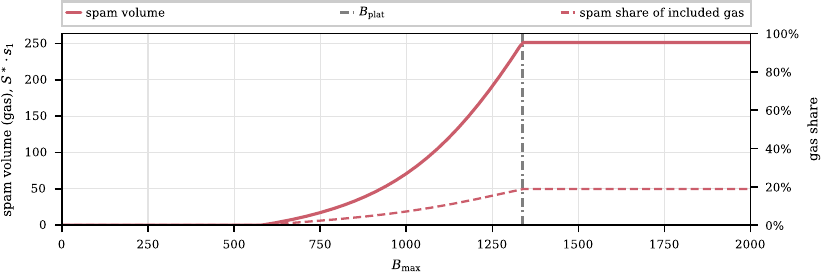}
    \caption{Equilibrium spam volume and spam share under exponential demand.
    Qualitatively, this yields similar results to linear demand.}
    \Description{Under exponential demand, a solid curve shows spam gas on the left axis and a dashed curve shows spam's share of included gas on the right axis, both against block size. Both are zero for small blocks, rise after entry becomes profitable, and flatten at the vertically marked plateau threshold.}
    \label{fig:app-exp-spam-volume}
\end{figure}

\parhead{User welfare.}
For exponential demand, consumer surplus has a simple form.
For included user gas $Q$, we have
$$
\int_{0}^{Q}
(P^{\exp}(q)-P^{\exp}(Q))dq
=
\frac{Q}{\lambda}.
$$
Thus, under exponential demand, user welfare depends linearly on the amount of included user gas
$
W_{\mathrm{user}}(Q)=\frac{Q}{\lambda}.
$
In the spam-free world,
$
W_{\mathrm{user}}^0(\bsmax,\gpmin)
=
\frac{Q_u^0(\bsmax,\gpmin)}{\lambda},
$
and in the spam world,
$
W_{\mathrm{user}}^\ast(\bsmax,\gpmin)
=
\frac{Q_u^\ast(\bsmax,\gpmin)}{\lambda}.
$
The user-welfare delta is therefore
$
\Delta W_{\mathrm{user}}(\bsmax,\gpmin)
=
\frac{Q_u^\ast(\bsmax,\gpmin)-Q_u^0(\bsmax,\gpmin)}{\lambda}.
$
As in the linear-demand model, this quantity is weakly negative and is zero whenever spam does not displace users.

\parhead{Validator revenue and externality.}
The validator-revenue and externality expressions are unchanged except that the clearing prices and user quantities are now computed using exponential demand.
In the spam-free world,
$$
R^0(\bsmax,\gpmin)=\gp^0 Q_u^0,
\qquad
E^0(\bsmax,\gpmin)=c_1\bsmax+c_2Q_u^0.
$$
In the spam world, let
$
G^\ast=Q_u^\ast+\spc\spm^\ast
$
denote total gas sold.
Then
$$
R^\ast(\bsmax,\gpmin)=\gp^\ast G^\ast,
\qquad
E^\ast(\bsmax,\gpmin)=c_1\bsmax+c_2G^\ast.
$$
The counterfactual deltas are
$$
\Delta R_{\mathrm{val}}=R^\ast-R^0,
\qquad
\Delta E=E^\ast-E^0=c_2(G^\ast-Q_u^0).
$$

\begin{figure*}[h]
    \centering
    \includegraphics{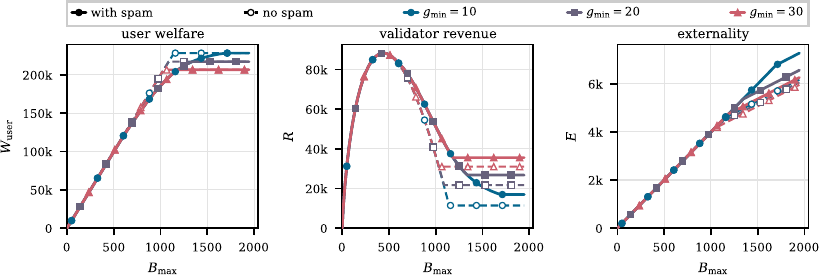}
    \caption{User welfare, validator revenue, and externality under exponential demand, with and without spam.
    }
    \Description{Three panels plot user welfare, validator revenue, and externality against block size under exponential demand. Colors and marker shapes distinguish minimum gas prices. Solid curves with filled markers show spam outcomes; dashed curves with hollow markers show no-spam outcomes. Welfare rises and flattens, revenue peaks and then declines to a plateau, and externality rises. Spam lowers welfare over intermediate capacities and raises revenue and externality where the corresponding curves separate.}
    \label{fig:app-exp-welfare}
\end{figure*}

\parhead{Marginal user share.}
The sufficient condition in \cref{prop:mmus-general} also holds for exponential demand.
Indeed, for $D^{\exp}(\gp)=D_0e^{-\lambda \gp}$, we have
$
(D^{\exp})'(\gp)=-\lambda D^{\exp}(\gp)
$
and
$
(D^{\exp})''(\gp)=\lambda^2D^{\exp}(\gp).
$
Therefore,
\begin{align*}
&\gp D^{\exp}(\gp)(D^{\exp})''(\gp)
+
2D^{\exp}(\gp)(D^{\exp})'(\gp)
-
2\gp\bigl((D^{\exp})'(\gp)\bigr)^2\\
&\qquad
=
-\lambda\bigl(2+\lambda\gp\bigr)\bigl(D^{\exp}(\gp)\bigr)^2
<0.
\end{align*}
Hence the marginal user share $m_{\mathrm{user}}$ is decreasing in $\bsmax$ under exponential demand as well.

\section{Spam Equilibrium under Gas-Used Charging}
\label{app:gas-used-charging}

In \cref{sec:random-ordering}, we assume that each spam transaction is charged for its gas limit $\spc$, independently of whether the transaction finds and executes a profitable opportunity.
Here, we consider an alternative TFM that instead charges each transaction for its realized gas use.
To isolate the effect of the charging rule, we retain the admission and block-capacity model from the main text: every spam transaction declares and reserves a gas limit $\spc$, so $\spc\spm$ units of block capacity are unavailable to genuine users.
However, a spam transaction that only probes and finds no opportunity is charged for $\spp$ gas, while a transaction that captures at least one opportunity is charged for $\spc$ gas.
We assume $0<\spp\le \spc$.
Thus, $\spc-\spp$ is the additional gas consumed when a probe proceeds to execution.

For a fixed spam volume $\spm$, the included user gas and clearing price remain
$
Q_u(\spm)
=
\min\{
\bigl(\bsmax-\spc\spm\bigr)^+,
D(\gpmin)
\}
$
and
$
\gp(\spm)
=
\max\{
\gpmin,
\frac{D_0-\bsmax+\spc\spm}{\beta}
\}.
$
As in the main text, the aggregate opportunity value is
$
r(\spm)=r_0 Q_u(\spm)/D_0.
$

\parhead{Expected utility with multiple opportunities.}
Recall from \cref{sec:random-ordering} that the expected revenue of each spam transaction is $r(\spm)/(\spm+1)$.
The expected gas use additionally depends on whether the spam transaction captures at least one opportunity.
Let $\ell$ denote the number of opportunities.
Consider the relative ordering of the $\ell$ opportunities and the $\spm$ spam transactions.
A fixed spam transaction captures at least one opportunity exactly when the preceding element in this restricted ordering is an opportunity-generating event.
Under random ordering, this occurs with probability $\ell/(\ell+\spm)$.
Its expected gas use is therefore
$
\spp+(\spc-\spp)\frac{\ell}{\ell+\spm}
=
\frac{\spp\spm+\ell\spc}{\ell+\spm}.
$

Consequently, the expected utility of each spam transaction is
$$
u_{\mathrm{used},\ell}(\spm)
=
\frac{r_0Q_u(\spm)}{D_0(\spm+1)}
-
\gp(\spm)
(
\spp+(\spc-\spp)\frac{\ell}{\ell+\spm}
).
$$
The competitive free-entry condition is
$$
\frac{r_0Q_u(\spm^\ast)}{D_0(\spm^\ast+1)}
=
\gp(\spm^\ast)
(
\spp+(\spc-\spp)\frac{\ell}{\ell+\spm^\ast}
).
$$
Unlike under gas-limit charging, aggregate opportunity value alone is not sufficient to characterize the equilibrium.
The number $\ell$ of separately generated opportunities also matters because it determines how often a spam transaction proceeds from probing to execution.

Let
$
Q_u^0=\min\{\bsmax,D(\gpmin)\}
$
and
$
\hat{\gp}
=
\max\{
\gpmin,\frac{D_0-\bsmax}{\beta}
\}
$
denote the user gas and clearing price in the spam-free world.
Evaluating the continuous free-entry condition as $\spm\downarrow0$ gives the no-entry condition
$
r_0Q_u^0/D_0\le \spc\hat{\gp}.
$

In the slack region, $Q_u^\ast=D(\gpmin)$ and $\gp^\ast=\gpmin$.
Define
$
A_{\min}:=r_0D(\gpmin)/(D_0\gpmin).
$
The zero-profit condition becomes
$$
\spp(\spm^\ast)^2
+
\bigl(\spp+\ell\spc-A_{\min}\bigr)\spm^\ast
+
\ell\bigl(\spc-A_{\min}\bigr)
=
0.
$$
When $A_{\min}>\spc$, the positive root is
$$
\spm_{\mathrm{slack},\ell}^{\mathrm{used}}
=
\frac{
A_{\min}-\spp-\ell\spc
+
\sqrt{
\bigl(A_{\min}-\spp-\ell\spc\bigr)^2
+
4M\spp\bigl(A_{\min}-\spc\bigr)
}
}{
2\spp
}.
$$
This candidate is a slack equilibrium when
$
D(\gpmin)+\spc\spm_{\mathrm{slack},\ell}^{\mathrm{used}}
\le \bsmax.
$
The corresponding plateau threshold is therefore
$
B_{\mathrm{plat},\ell}^{\mathrm{used}}
=
D(\gpmin)+\spc\spm_{\mathrm{slack},\ell}^{\mathrm{used}}.
$

In the congested region, the block is full, so
$
Q_u^\ast=\bsmax-\spc\spm^\ast
$
and
$
\gp^\ast=(D_0-\bsmax+\spc\spm^\ast)/\beta.
$
Let $\Delta=D_0-\bsmax$.
Substituting these expressions into the free-entry condition gives
$$
\beta r_0
\bigl(\bsmax-\spc\spm^\ast\bigr)
\bigl(\ell+\spm^\ast\bigr)
=
D_0
\bigl(\spm^\ast+1\bigr)
\bigl(\Delta+\spc\spm^\ast\bigr)
\bigl(\spp\spm^\ast+\ell\spc\bigr).
$$
This is generally a cubic equation in $\spm^\ast$.
The congested equilibrium is the economically relevant nonnegative root satisfying $\spc\spm^\ast\le\bsmax$ and $\gp^\ast>\gpmin$, and can be computed numerically.

Combining the regimes, spam does not enter when
$
r_0Q_u^0/D_0\le\spc\hat{\gp}.
$
Otherwise, the equilibrium is $\spm_{\mathrm{slack},\ell}^{\mathrm{used}}$ when the slack feasibility condition holds, and is the relevant root of the congested equation above when it does not.

When $\spp=\spc$, every spam transaction pays for the same amount of gas, whether or not it captures an opportunity.
The expected gas cost then becomes $\spc\gp(\spm)$, the dependence on $\ell$ disappears, and the equilibrium reduces to the gas-limit model in \cref{subsec:equilibrium-random}.
When $\spp<\spc$, failed probes are cheaper, so gas-used charging weakens the effect created by transaction fees and sustains a larger spam volume in the slack region.

\section{Priority Fee Ordering}
\subsection{Equilibrium Analysis}
\label{app:pfo-equilibrium}

In this appendix, we generalize the $n=2$ model from the main body to an arbitrary number $n$ of sub-blocks.
Unlike the main body, which uses a free parameter $v$ for the top-block demand share, here we assume an even partition of the original linear demand curve across sub-blocks.

Let $C=\bsmax/n$ denote the capacity of each sub-block.
Let $S_i$ denote the number of spam transactions in sub-block $i$, and let each spam transaction reserve $\spc$ gas.
For $i=1,\dots,n$, define the inverse demand of sub-block $i$ by
$$
P_i(Q)=\frac{\frac{n-i+1}{n}D_0-Q}{\beta},
\qquad
0\le Q\le \frac{D_0}{n}.
$$
Equivalently, the direct demand is
$$
D_i(\gp)
=
(
\min\{
\frac{D_0}{n},
\frac{n-i+1}{n}D_0-\beta\gp
\}
)^+.
$$
These demand slices sum to the original linear demand:
$$
\sum_{i=1}^{n}D_i(\gp)=D_0-\beta\gp.
$$
Thus, each sub-block is associated with a distinct slice of the original valuation distribution, and users from earlier slices do not spill into later sub-blocks.

For a candidate spam profile $\mathbf{S}=(S_1,\dots,S_n)$ and a candidate block-clearing price $\gpu$, define the included user gas in sub-block $i$ by
$$
Q_i(\mathbf{S};\gpu)
=
(
\min\{
C-S_i\spc,
D_i(\gpu)
\}
)^+.
$$
The clearing price of sub-block $i$ is then
$$
\gp_i(\mathbf{S};\gpu)
=
\max\{
\gpu,
P_i(Q_i(\mathbf{S};\gpu))
\}.
$$
Because the demand slices are nested, these prices are automatically weakly decreasing down the block.

The total included user gas is
$$
Q_u(\mathbf{S};\gpu)=\sum_{i=1}^{n}Q_i(\mathbf{S};\gpu).
$$
As in the main text, we allow multiple non-interacting opportunities.
The aggregate value of the opportunities created by user gas in sub-block $i$ is assumed to be proportional to the amount of user gas in that sub-block, and is given by
$$
R_i(\mathbf{S};\gpu)
:=
\frac{r_0}{D_0}Q_i(\mathbf{S};\gpu).
$$
This is the sub-block analogue of the assumption that total opportunity value scales with included user gas.
Only this aggregate value matters for the spam equilibrium.
Indeed, if sub-block $i$ contains opportunities $\ell$ with values $r_\ell$, then for any fixed spam volume the expected captured value is the sum of the expected captured values of the individual opportunities.
Thus, by linearity of expectation, multiple non-interacting opportunities in the same sub-block are equivalent to a single aggregate opportunity of value $R_i(\mathbf{S};\gpu)$.

\parhead{Top and second top sub-blocks.}
Given $S_1$ spam transactions in the top sub-block, each opportunity created in that sub-block is captured by spam in the same sub-block with probability $S_1/(S_1+1)$.
This is the same relative-ordering argument as in \cref{subsec:equilibrium-random}, applied to each opportunity separately.
By linearity of expectation over all opportunities created in the top sub-block, the aggregate expected utility of spam in the top sub-block is
$$
U_1(S_1)
=
\frac{r_0}{D_0}Q_1(S_1;\gpu)\frac{S_1}{S_1+1}
-
S_1\spc\,\gp_1(S_1;\gpu).
$$
Setting $U_1(S_1^\ast)=0$ gives the equilibrium spam volume in the top sub-block, subject to the capacity constraint $S_1^\ast\spc\le C$.

For sub-block $2$, spam can capture two sources of value.
First, it can capture opportunities created in sub-block $2$ itself.
The aggregate expected value of these opportunities captured in sub-block $2$ is
$$
\frac{r_0}{D_0}Q_2(S_2;\gpu)\frac{S_2}{S_2+1}.
$$
Second, opportunities created in sub-block $1$ may fail to be captured there.
Each such opportunity survives sub-block $1$ with probability $1/(S_1^\ast+1)$.
The aggregate value of surviving opportunities from sub-block $1$ is therefore
$$
\frac{r_0}{D_0}\frac{Q_1^\ast}{S_1^\ast+1}.
$$
If $S_2>0$, these surviving opportunities reach sub-block $2$ and are captured by the first spam transaction in that sub-block.
Thus, for $S_2>0$, the aggregate expected utility of spam in the second sub-block is
$$
U_2(S_2)
=
\frac{r_0}{D_0}
(
Q_2(S_2;\gpu)\frac{S_2}{S_2+1}
+
\frac{Q_1^\ast}{S_1^\ast+1}
)
-
S_2\spc\,\gp_2(S_2;\gpu).
$$
When $S_2=0$, no spam transaction in the second sub-block is present to capture surviving opportunities, so we take $U_2(0)=0$.
Solving $U_2(S_2^\ast)=0$ gives the second-sub-block equilibrium, subject to the capacity constraint $S_2^\ast\spc\le C$.

\parhead{General sub-block.}
For sub-block $i\ge 3$, suppose that $S_1^\ast,\dots,S_{i-1}^\ast$ have already been computed.
The aggregate expected revenue of spam in sub-block $i$ again has two components: opportunities created in sub-block $i$ and captured there, and opportunities created in earlier sub-blocks that remain unclaimed until sub-block $i$.

Let $h_i$ be the index of the last earlier sub-block with positive spam, that is,
$$
h_i=\max\{j<i:S_j^\ast>0\},
$$
with the convention $h_i=0$ if no earlier sub-block has positive spam.
Let
$$
k_i=i-1-h_i
$$
be the number of contiguous zero-spam sub-blocks immediately before sub-block $i$.
If opportunities are created in any of those $k_i$ zero-spam sub-blocks, then they reach sub-block $i$ alive and are captured there by any positive spam.
The aggregate value of those opportunities is therefore
$$
\frac{r_0}{D_0}\sum_{j=h_i+1}^{i-1}Q_j^\ast.
$$
If $h_i\ge 1$, then there is one additional spillover source: opportunities created in sub-block $h_i$ may fail to be captured there.
Each such opportunity survives with probability $1/(S_{h_i}^\ast+1)$.
Since all sub-blocks between $h_i$ and $i$ have zero spam, those surviving opportunities then reach sub-block $i$ alive.
This contributes
$$
\frac{r_0}{D_0}\frac{Q_{h_i}^\ast}{S_{h_i}^\ast+1}.
$$
Putting these terms together, for $S_i>0$, the aggregate expected utility of spam in sub-block $i$ is
$$
U_i(S_i)
=
\frac{r_0}{D_0}
(
Q_i(S_i;\gpu)\frac{S_i}{S_i+1}
+
\sum_{j=h_i+1}^{i-1}Q_j^\ast
+
\mathbf{1}[h_i\ge 1]\frac{Q_{h_i}^\ast}{S_{h_i}^\ast+1}
)
-
S_i\spc\,\gp_i(S_i;\gpu).
$$
When $S_i=0$, no spam transaction in sub-block $i$ is present to capture surviving opportunities, so we take $U_i(0)=0$.
We solve $U_i(S_i^\ast)=0$ iteratively for $i=1,\dots,n$.
If $U_i(S_i)$ remains nonnegative all the way up to the capacity boundary $S_i=C/\spc$, then the equilibrium for that sub-block is the corner solution $S_i^\ast=C/\spc$.

This gives a spam profile $S_1^\ast,\dots,S_n^\ast$.
The equilibrium is then closed by the fixed-point condition
$$
\gpu^\ast
=
\max\{
\gpmin,
P_n(Q_n(S_n^\ast;\gpu^\ast))
\}.
$$
The total equilibrium spam volume is
$
S^\ast=\sum_{i=1}^n S_i^\ast,
$
and the total included user gas is
$
Q_u^\ast=\sum_{i=1}^{n}Q_i^\ast.
$

\subsection{Derivation of Welfare Metrics}
\label{app:pfo-metrics}

We now record the outcome metrics for approximate PFO.
We first give the expressions for the two-sub-block model used in the main body, where the first sub-block has capacity $C_1=v\bsmax$ and the second sub-block has capacity $C_2=(1-v)\bsmax$.
We then state the corresponding expressions for the general $n$-sub-block model from \cref{app:pfo-equilibrium}.

\parhead{Two-sub-block model.}
Recall that the first sub-block is associated with the upper slice of the demand curve, with inverse demand $P_1^{(v)}$, while the second sub-block is associated with the lower slice, with inverse demand $P_2^{(v)}$.
Changing $v$ therefore changes both the capacity allocated to the early block region and the distribution of user valuations across sub-blocks.

\parhead{User welfare.}
Because the two sub-blocks correspond to different segments of the original demand curve, welfare is computed separately in each segment.
In the spam world, users in the first sub-block pay $\gp_1^\ast$, while users in the second sub-block pay $\gpu^\ast$.
Therefore,
$$
W_{\text{user}}(\bsmax)
=
\int_{0}^{Q_1^\ast}
(P_1^{(v)}(q)-\gp_1^\ast)\,dq
+
\int_{0}^{Q_2^\ast}
(P_2^{(v)}(q)-\gpu^\ast)\,dq.
$$
The no-spam benchmark is obtained by setting $S_1=S_2=0$ and recomputing the two clearing prices and user quantities.

\parhead{Validator revenue.}
Each unit of gas in sub-block $1$ pays $\gp_1^\ast$, while each unit of gas in sub-block $2$ pays $\gpu^\ast$.
The total gas sold in each sub-block is the sum of user gas and spam gas.
Therefore,
$$
R(\bsmax)
=
\gp_1^\ast\cdot\bigl(Q_1^\ast+S_1^\ast\spc\bigr)
+
\gpu^\ast\cdot\bigl(Q_2^\ast+S_2^\ast\spc\bigr).
$$

\parhead{Externality.}
As in \cref{subsec:metrics}, we model the cost of supporting larger blocks as a capacity cost plus a per-gas execution cost.
In the spam world,
$$
E(\bsmax)
=
c_1\bsmax
+
c_2\bigl(Q_1^\ast+Q_2^\ast+\spc S_1^\ast+\spc S_2^\ast\bigr).
$$
The no-spam counterfactual is obtained by setting $S_1=S_2=0$.

\begin{figure*}[t]
    \centering
    \includegraphics{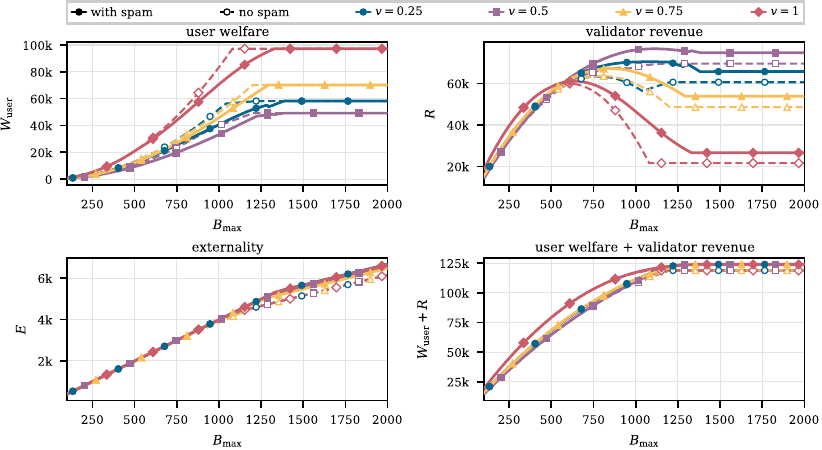}
    \caption{User welfare, validator revenue, externality, and $W_{\mathrm{user}}+R$ under approximate priority fee ordering in the two-sub-block model.
    The solid curves show the spam world and the dashed curves show the no-spam counterfactual.
    Each color corresponds to a different value of $v$.}
    \Description{Four panels plot user welfare, validator revenue, externality, and the sum of welfare and revenue against block size. Colors and marker shapes distinguish first-sub-block allocations; solid curves denote spam and dashed curves no spam. Welfare rises to allocation-dependent plateaus, revenue generally peaks before flattening, and externality increases. Welfare plus revenue approaches similar plateaus across allocations, with spam curves above the corresponding no-spam curves at large capacities.}
    \label{fig:pfo-welfare}
\end{figure*}

\Cref{fig:pfo-welfare} compares the spam and no-spam outcomes under approximate PFO.

\parhead{General $n$-sub-block model.}
For completeness, we also record the corresponding metrics for the general $n$-sub-block model from \cref{app:pfo-equilibrium}.
In that model, each sub-block corresponds to its own slice of the original demand curve.
Users in sub-block $i$ have inverse demand $P_i$ and pay price $\gp_i^\ast=\gp_i(S_i^\ast;\gpu^\ast)$.
Their consumer surplus is therefore
$$
\int_{0}^{Q_i^\ast}(P_i(q)-\gp_i^\ast)\,dq.
$$
Summing over sub-blocks gives total user welfare:
$$
W_{\text{user}}(\bsmax)
=
\sum_{i=1}^{n}
\int_{0}^{Q_i^\ast}(P_i(q)-\gp_i^\ast)\,dq.
$$

Each unit of gas in sub-block $i$ pays $\gp_i^\ast$.
The total gas sold in sub-block $i$ is $Q_i^\ast+S_i^\ast\spc$.
Therefore, validator revenue is
$$
R(\bsmax)
=
\sum_{i=1}^{n}\gp_i^\ast\cdot\bigl(Q_i^\ast+S_i^\ast\spc\bigr).
$$

Using the same system-cost model as above, the externality is
$$
E(\bsmax)
=
c_1\bsmax+c_2\sum_{i=1}^{n}\bigl(Q_i^\ast+S_i^\ast\spc\bigr).
$$
The no-spam benchmark is obtained by setting $S_i=0$ for every sub-block and recomputing the induced prices and user quantities.

\section{Spam Volume Share under Demand Scaling with Approximate Priority Fee Ordering}
\label{app:demand-scaling-pfo}

We ask whether priority fee ordering changes the scaling picture.
To keep the comparison aligned with the random-ordering analysis in the main body, we use the same block-size benchmark.
For each value of $\lambda$, we set $\bsmax=\bs_{\mathrm{plat},\lambda}$, where $\bs_{\mathrm{plat},\lambda}$ is the random-ordering plateau size defined in \cref{sec:demand-scaling}.
Thus, total provisioned capacity is held fixed at the same benchmark level, while the allocation of block capacity across execution positions changes according to the two-sub-block PFO model.

For a given value of $v$, let $\spm_\lambda^{(v)}$ denote the equilibrium spam volume in the two-sub-block PFO model, and let $Q_{u,\lambda}^{(v)}$ denote the corresponding included user gas.
Recall that $v$ is the fraction of both block capacity and user valuation mass assigned to the first, higher-priority sub-block.
We define the spam share of included gas as
$$
\rho_{\mathrm{spam}}^{(v)}(\lambda)
=
\frac{\spc\,\spm_\lambda^{(v)}}{\spc\,\spm_\lambda^{(v)}+Q_{u,\lambda}^{(v)}}.
$$
The case $v=1$ collapses to a single sub-block and is therefore equivalent to the random-ordering benchmark.

Figure~\ref{fig:demand-scaling-pfo} compares the two-sub-block PFO outcomes for $v\in\{0.25,0.5,0.75,1.0\}$.
We fix $\gpmin=20$ and size the block using $\bs_{\mathrm{plat},\lambda}$.
The figure shows that PFO can reduce the spam share relative to the $v=1$ benchmark.

The qualitative scaling pattern remains similar to the random-ordering case.
Across the curves, the spam share approaches a positive plateau as $\lambda$ grows.
Thus, approximate PFO changes the level of spam pressure, but does not change the basic conclusion that linear opportunity scaling preserves a nontrivial spam share.
The role of PFO is therefore to reduce the share of spam in favorable parameter regimes.

\begin{figure}[t]
    \centering
    \includegraphics{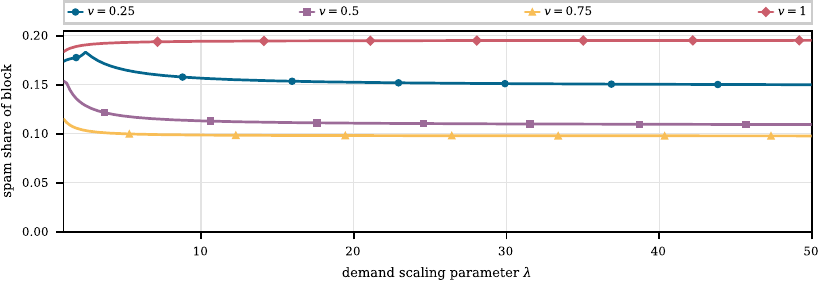}
    \caption{Spam volume share under demand scaling with approximate priority fee ordering.
    The x-axis is the demand-scaling parameter $\lambda$.
    The y-axis is the fraction of included gas consumed by spam.
    All curves are evaluated at the same block-size benchmark $\bsmax=\bs_{\mathrm{plat},\lambda}$.
    The curves show the two-sub-block PFO outcome for $v\in\{0.25,0.5,0.75,1.0\}$, where $v$ is the fraction of both block capacity and user valuation mass assigned to the first sub-block.
    The case $v=1$ collapses to the random-ordering benchmark.}
    \Description{Four curves plot spam share against the demand-scaling parameter. The random-ordering curve rises slightly toward the highest plateau. The approximate-priority-ordering curves generally decrease toward lower positive plateaus, with a small initial rise for the 0.25 allocation. At large demand, the 0.75 allocation has the lowest spam share, followed by 0.5, 0.25, and random ordering.}
    \label{fig:demand-scaling-pfo}
\end{figure}

\section{Proofs}
\label{app:proofs}

\subsection{Proof of \texorpdfstring{\cref{prop:welfare-peak}}{Proposition (Welfare Loss Peak)}}

We work with the linear demand $D(\gp)=D_0-\beta\gp$, inverse demand $P(Q)=(D_0-Q)/\beta$, and $Q_{\min}:=D_0-\beta\gpmin$.
Consumer surplus is $W(Q)=Q^2/(2\beta)$.
Recall that the opportunity value is now endogenous and equal to $r=r_0 Q_u/D_0$.
If $r_0 Q_{\min}/D_0\le \spc\gpmin$, then spam does not enter for any $\bsmax$, so $\Delta W_{\text{user}}\equiv 0$ and the claim is trivial.
We therefore assume that entry occurs.

In the congested regime, let $Q:=Q_u^\ast$.
Free entry, block fullness, and the linear inverse demand imply $\frac{r_0 Q}{D_0(\spm^\ast+1)}=\spc\gp^\ast$, $\bsmax=Q+\spc\spm^\ast$, and $\gp^\ast=\frac{D_0-Q}{\beta}$.
Eliminating $\spm^\ast$ and $\gp^\ast$ gives
$$
\beta r_0 Q = D_0(D_0-Q)(\bsmax-Q+\spc).
$$
Differentiating implicitly with respect to $\bsmax$ yields
$$
Q'(\bsmax)
=
\frac{D_0(D_0-Q)}{\beta r_0 + D_0(\bsmax+\spc+D_0-2Q)}.
$$
Hence $0<Q'(\bsmax)<1$, since $Q<\bsmax$ in the spam world and the denominator exceeds the numerator by
$$
\beta r_0 + D_0(\bsmax+\spc-Q)>0.
$$

\medskip
\noindent\textbf{Right of $Q_{\min}$ (loss shrinks).}
The spam-free world is slack, so $W^0=Q_{\min}^2/(2\beta)$ is constant.
For $Q_{\min}<\bsmax<B_{\mathrm{plat}}$, the spam world is congested and
$$
W^\ast=\frac{Q(\bsmax)^2}{2\beta},
\qquad
\frac{dW^\ast}{d\bsmax}=\frac{Q(\bsmax)Q'(\bsmax)}{\beta}>0.
$$
Thus $\Delta W_{\text{user}}=W^\ast-W^0$ increases toward $0$ as $\bsmax$ grows.
Once $\bsmax\ge B_{\mathrm{plat}}$, the spam world is also slack and $\Delta W_{\text{user}}=0$.

\medskip
\noindent\textbf{Left of $Q_{\min}$ (loss grows).}
Below the spam-entry threshold, $\Delta W_{\text{user}}=0$.
Once spam enters and $\bsmax<Q_{\min}$, both worlds are congested:
$$
W^0=\frac{\bsmax^2}{2\beta},
\qquad
W^\ast=\frac{Q(\bsmax)^2}{2\beta}.
$$
Therefore
$$
\Delta W_{\text{user}}
=
W^\ast-W^0
=
\frac{Q(\bsmax)^2-\bsmax^2}{2\beta},
$$
and
$$
\frac{d}{d\bsmax}\Delta W_{\text{user}}
=
\frac{Q(\bsmax)Q'(\bsmax)-\bsmax}{\beta}
<0,
$$
because $Q(\bsmax)<\bsmax$ and $0<Q'(\bsmax)<1$.
So the welfare loss becomes more negative as $\bsmax$ increases up to $Q_{\min}$.

Combining the two regions, $\Delta W_{\text{user}}(\bsmax)$ decreases on the left of $Q_{\min}$ and increases on the right of $Q_{\min}$.
Therefore it is most negative at $\bsmax=Q_{\min}$.
\qed

\subsection{Proof of \texorpdfstring{\cref{prop:mmus-general}}{Proposition (MMUS)}}
\begin{proof}
In the entry-and-congested region, free entry implies
$$
\frac{r_0 D(\gp^\ast)}{D_0(\spm^\ast+1)}=\spc\,\gp^\ast,
$$
so
$$
\spm^\ast=\frac{r_0 D(\gp^\ast)}{D_0\spc\,\gp^\ast}-1.
$$
Because the block is congested, it is full, so
$$
D(\gp^\ast)+\spc\,\spm^\ast=\bsmax.
$$
Substituting the expression for $\spm^\ast$ gives
$$
r_0 D(\gp^\ast)=D_0\gp^\ast\bigl(\bsmax-D(\gp^\ast)+\spc\bigr).
$$

Differentiating both sides with respect to $\bsmax$ yields
$$
r_0 D'(\gp^\ast)\frac{\partial \gp^\ast}{\partial \bsmax}
=
D_0[
\frac{\partial \gp^\ast}{\partial \bsmax}\bigl(\bsmax-D(\gp^\ast)+\spc -\gp^\ast D'(\gp^\ast)\bigr)
+
\gp^\ast
].
$$
Rearranging and using
$$
\bsmax-D(\gp^\ast)+\spc=\frac{r_0 D(\gp^\ast)}{D_0\gp^\ast}
$$
from the equilibrium condition above gives
$$
\frac{\partial \gp^\ast}{\partial \bsmax}
=
-\frac{D_0(\gp^\ast)^2}{r_0 D(\gp^\ast)-\gp^\ast(r_0+D_0\gp^\ast)D'(\gp^\ast)}.
$$

Now use $Q_u^\ast=D(\gp^\ast)$, so
$$
m_{\text{user}}
=
\frac{\partial Q_u^\ast}{\partial \bsmax}
=
-\frac{D_0(\gp^\ast)^2D'(\gp^\ast)}
{r_0 D(\gp^\ast)-\gp^\ast(r_0+D_0\gp^\ast)D'(\gp^\ast)}.
$$

Differentiating $m_{\text{user}}$ with respect to $\gp^\ast$ gives
$$
\frac{\partial m_{\text{user}}}{\partial \gp^\ast}
=
-\frac{D_0\gp^\ast r_0\Bigl(\gp^\ast D(\gp^\ast)D''(\gp^\ast)+2D(\gp^\ast)D'(\gp^\ast)-2\gp^\ast(D'(\gp^\ast))^2\Bigr)}
{\Bigl(r_0 D(\gp^\ast)-\gp^\ast(r_0+D_0\gp^\ast)D'(\gp^\ast)\Bigr)^2}.
$$
Multiplying by $\partial \gp^\ast/\partial \bsmax$ and using the formula above yields
$$
\frac{\partial m_{\text{user}}}{\partial \bsmax}
=
\frac{D_0^2(\gp^\ast)^3 r_0\Bigl(\gp^\ast D(\gp^\ast)D''(\gp^\ast)+2D(\gp^\ast)D'(\gp^\ast)-2\gp^\ast(D'(\gp^\ast))^2\Bigr)}
{\Bigl(r_0 D(\gp^\ast)-\gp^\ast(r_0+D_0\gp^\ast)D'(\gp^\ast)\Bigr)^3}.
$$
The denominator is positive, so the sign is determined by
$$
\gp^\ast D(\gp^\ast)D''(\gp^\ast)+2D(\gp^\ast)D'(\gp^\ast)-2\gp^\ast(D'(\gp^\ast))^2.
$$
Thus, whenever
$$
\gp^\ast D(\gp^\ast)D''(\gp^\ast)+2D(\gp^\ast)D'(\gp^\ast)-2\gp^\ast(D'(\gp^\ast))^2<0,
$$
we have
$$
\frac{\partial m_{\text{user}}}{\partial \bsmax}<0.
$$

For the linear demand curve $D(\gp)=D_0-\beta\gp$, we have $D'(\gp)=-\beta$ and $D''(\gp)=0$, so
$$
\gp D(\gp)D''(\gp)+2D(\gp)D'(\gp)-2\gp(D'(\gp))^2
=
-2\beta D_0
<0.
$$
Therefore, the sufficient condition above holds for the linear demand function.
\end{proof}

\begin{remark}
\label{rem:mmus-general-demand}
The condition in \cref{prop:mmus-general} does not hold for every decreasing demand curve.
For example, let
$$
D(\gp)=A\exp(-\frac{1-(1+\gp)^{-2}}{2}).
$$
Then $D'(\gp)=-D(\gp)/(1+\gp)^3<0$, so the demand curve is strictly decreasing.
A direct calculation gives
\[
\gp D(\gp)D''(\gp)+2D(\gp)D'(\gp)-2\gp(D'(\gp))^2
=
A^2 e^{-1+(1+\gp)^{-2}}
\frac{\gp^3-4\gp-2}{(1+\gp)^6}.
\]
This expression is positive for all sufficiently large $\gp$.
Therefore, the sufficient condition in \cref{prop:mmus-general} can fail even for a smooth strictly decreasing demand curve.
\end{remark}

\begin{remark}[The condition that the slope of $m_{\text{user}}$ decreasing in $\bsmax$]
Let
\[
A(\gp)=r_0 D(\gp)-\gp(r_0+D_0\gp)D'(\gp),
\]
and
\[
F(\gp)=\gp D(\gp)D''(\gp)+2D(\gp)D'(\gp)-2\gp(D'(\gp))^2.
\]
We have
\[
\frac{\partial m_{\text{user}}}{\partial \bsmax}
=
\frac{D_0^2(\gp^\ast)^3 r_0\,F(\gp^\ast)}{A(\gp^\ast)^3}.
\]
Since
\[
\frac{\partial \gp^\ast}{\partial \bsmax}<0,
\]
the condition for \(\partial m_{\text{user}}/\partial \bsmax\) to be non-increasing is
\[
\frac{\mathrm d}{\mathrm d\gp}
(
\frac{\gp^3 F(\gp)}{A(\gp)^3}
)
\ge 0.
\]
It holds for the linear demand function.
\end{remark}
\begin{figure*}[t]
    \centering
    \begin{subfigure}[h]{0.49\linewidth}
        \centering
        \includegraphics{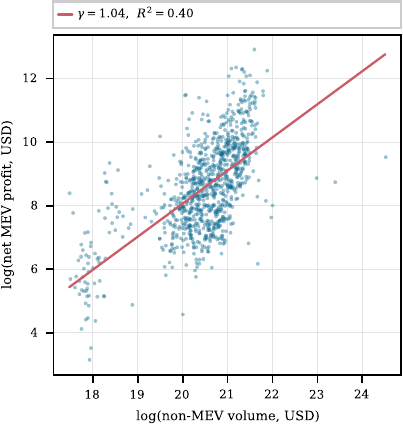}
        \caption{Base. The estimated elasticity is $\hat{\gamma}=1.04$ with $R^2=0.40$.}
        \label{fig:mev-elasticity-base}
    \end{subfigure}
    \hfill
    \begin{subfigure}[h]{0.49\linewidth}
        \centering
        \includegraphics{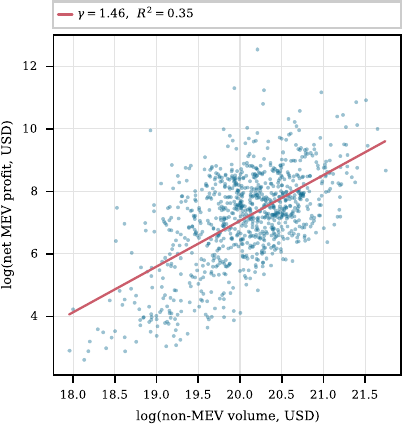}
        \caption{Arbitrum. The estimated elasticity is $\hat{\gamma}=1.46$ with $R^2=0.35$.}
        \label{fig:mev-elasticity-arbitrum}
    \end{subfigure}
    \caption{Daily net cyclic arbitrage profit versus daily non-MEV trading volume on log-log scales.}
    \Description{Two scatterplots show the logarithm of daily net cyclic-arbitrage profit against the logarithm of daily non-MEV trading volume, with Base on the left and Arbitrum on the right. Each point represents a day. Both point clouds trend upward with substantial vertical scatter, and each panel includes an upward-sloping fitted regression line. The fitted line is steeper for Arbitrum.}
    \label{fig:mev-elasticity}
\end{figure*}

\subsection{Proof of \texorpdfstring{\cref{prop:mu-user-gmin}}{Proposition (User Share under Lower Gas Price Floor)}}
\begin{proof}
In the no-entry region, $\spm^\ast=0$, so lowering the gas price floor changes only user inclusion.
Therefore the full marginal increase in used capacity goes to users, and $\mu_{\text{user}}=1$.
When the floor binds, no entry means
$$
\frac{r_0 D(\gpmin)}{D_0}\le \spc\,\gpmin.
$$
Under the linear demand function $D(\gp)=D_0-\beta\gp$, this is equivalent to
$$
\gpmin\ge \frac{r_0 D_0}{D_0\spc+\beta r_0}.
$$
Together with the condition that the floor binds, $\gpmin\ge \frac{D_0-\bsmax}{\beta}$, this gives the first case.

In the slack-with-spam region, we have
$$
Q_u^\ast(\bsmax,\gpmin)=D(\gpmin)=D_0-\beta\gpmin
$$
and
$$
\spc\,\spm^\ast(\bsmax,\gpmin)
=
\spc(\frac{r_0 D(\gpmin)}{D_0\spc\,\gpmin}-1)
=
\frac{r_0 D(\gpmin)}{D_0\gpmin}-\spc.
$$
Under the linear demand function, this simplifies to
$$
\spc\,\spm^\ast(\bsmax,\gpmin)
=
\frac{r_0(D_0-\beta\gpmin)}{D_0\gpmin}-\spc
=
\frac{r_0}{\gpmin}-\frac{\beta r_0}{D_0}-\spc.
$$
Differentiating with respect to $\gpmin$ gives
$$
\frac{\partial Q_u^\ast}{\partial \gpmin}=-\beta
\qquad\text{and}\qquad
\frac{\partial (\spc\,\spm^\ast)}{\partial \gpmin}=-\frac{r_0}{\gpmin^2}.
$$
Substituting these derivatives into the definition of $\mu_{\text{user}}$ yields
$$
\mu_{\text{user}}(\bsmax,\gpmin)
=
\frac{\beta}{\beta+r_0/\gpmin^2}
=
\frac{\beta\,\gpmin^2}{\beta\,\gpmin^2+r_0}.
$$
Differentiating this expression gives
$$
\frac{\partial \mu_{\text{user}}(\bsmax,\gpmin)}{\partial \gpmin}
=
\frac{2\beta r_0\,\gpmin}{(\beta\gpmin^2+r_0)^2}>0.
$$
Therefore $\mu_{\text{user}}$ increases with $\gpmin$, or equivalently decreases as the gas price floor is lowered.

Finally, in the congested-with-spam region, the equilibrium price is already above the gas price floor.
A small change in $\gpmin$ does not change the local equilibrium allocation, so both the user-gas increment and the spam-gas increment are zero.
Thus the denominator in the definition of $\mu_{\text{user}}$ vanishes, and the local share is not informative in that region.
\end{proof}

\section{Data and methodology}
\label{app:data}

We collect 912 daily observations per chain from January~2024 through June~2026 using Dune Analytics. To classify spam, we build on the methodology employed by Flashbots~\cite{miller-mev-scaling-2025} and are consistent with the optimistic probing behavior described by Solmaz et al.~\cite{solmaz2025optimistic}. We identify contracts that repeatedly probe for MEV opportunities that never trigger meaningful state changes, which we detect through the absence of any ERC-20 token transfer.

The key observable is the ERC-20 transfer event. A successful MEV action, whether an arbitrage swap or a liquidation, emits at least one transfer, whereas a failed probe emits none. A probe can fail either because the transaction reverted or because it succeeded but did nothing, having read state, found no opportunity, and returned early.

Using on-chain traces, we capture the two forms of probing introduced in \cref{sec:intro}, (cyclic) arbitrage and liquidation. Each form is anchored on a call that reveals intent. For arbitrage, this is a successful call to a contract that Dune's \texttt{dex.trades} table lists as a trading venue, such as a pool or a vault. For liquidation, it is a staticcall to one of a curated set of Chainlink price-feed proxies, which counts whether or not the transaction later reverts.

For each form, a contract becomes a candidate in a month if at least 10 of its transactions in that month make the anchoring call and fewer than 50\% of these transactions emit an ERC-20 transfer. Among the candidates, we keep the 100 with the most gas in anchoring transactions per form and month. A retained contract counts as spam in that month only if, in addition, fewer than 50\% of all transactions sent to it in that month emit an ERC-20 transfer. Finally, we drop contracts that we identify as not MEV bots. We screen the 100 largest flagged contracts per chain, as described below and listed in \cref{tab:excluded}. These last two filters only remove contracts, and removed contracts are not replaced. 

Spam gas is the gas of all transactions sent to a spam contract in a month in which it is spam. Contracts that qualify under both forms are counted once in the total but attributed to each form.
Non-spam gas is the total gas minus spam gas.
The scripts that we use to gather the data can be found on \href{https://doi.org/10.5281/zenodo.23445675}{this page}.

\emph{Contract exclusion.} Some flagged contracts show the same pattern without doing spam MEV. Oracle nodes, for example, read price feeds and DEX pools in most of their transactions to post prices, but rarely move tokens. We check the verified source code and sampled transactions of the 100 largest flagged contracts per chain and exclude those that are not MEV bots, together with three more nodes of one of their protocols. \Cref{tab:excluded} lists these 17 contracts, which account for 2.6\% (Base) and 6.1\% (Arbitrum) of the classifier's spam gas. We do not screen smaller contracts, but a similar rate of non-bots among them would change the aggregates only slightly.

\paragraph{Classification accuracy.}
Our results rest on aggregate series, such as daily spam gas and its share of total gas, and on how these change around parameter changes. A misclassified contract matters only through its share of aggregate gas, and the classifier keeps only the largest contracts per form and month. We therefore check the largest contracts individually.

\emph{Audit of the largest contracts.} We audit the 20 largest spam contracts per chain, which account for 28\% of spam gas on Base and 33\% on Arbitrum. For each contract, we pick 400 transactions from at least 20 random hours in the months in which it counts as spam and record whether each succeeded and which tokens it moved. Most sampled transactions succeed without moving any token. Of the transactions that move tokens, all 60 on Base and 514 of 860 on Arbitrum are cyclic arbitrages, which return at least the tokens they put in. Another 342 on Arbitrum are swaps that end in a different token than they start with, and the remaining 4 are other transfers. The largest contracts thus behave as our rule assumes. Thus, they mostly check for an opportunity without acting on it, and when they do act, they almost always trade, mostly as cyclic arbitrage. This is consistent with spam MEV.

\emph{Smaller contracts.} The classifier also likely misses small and short-lived spam bots. It keeps only the top 100 contracts per form and month. Bots below these thresholds or with other strategies are not counted, so our measures likely understate spam overall. However, Spam gas is concentrated in few contracts. The 100 largest flagged contracts account for 55\% (Base) and 78\% (Arbitrum) of flagged gas.

\begin{table}[h]
  \centering
  \scriptsize
\begin{tabular}{@{}lllrr@{}}
\toprule
Contract & Identified as & Form & Months & \shortstack{Share of\\spam gas} \\
\midrule
\multicolumn{5}{@{}l}{\textit{Base}} \\
\texttt{0x5ff137d4b0fdcd49dca30c7cf57e578a026d2789} & ERC-4337 EntryPoint v0.6 & liquidation & 2 & 0.95\% \\
\texttt{0x0000000071727de22e5e9d8baf0edac6f37da032} & ERC-4337 EntryPoint v0.7 & liquidation & 2 & 0.72\% \\
\texttt{0xbfdd28226d6f0c808beeda9a58b69ba0b5005d09} & Supra random-number generator & DEX + liquidation & 25 & 0.35\% \\
\texttt{0x6a000f20005980200259b80c5102003040001068} & Paraswap router & liquidation & 6 & 0.30\% \\
\texttt{0x608667dd45d9d0d93e1da5220e2a712b8d52f8fb} & Betting application (ERC-404 token) & DEX & 10 & 0.23\% \\
 & & & Total & 2.6\% \\
\midrule
\multicolumn{5}{@{}l}{\textit{Arbitrum}} \\
\texttt{0x4213d42a5a6bd38ef9a166b179d1f360ff536d39} & Application proxy & liquidation & 10 & 2.06\% \\
\texttt{0x11d62807dae812a0f1571243460bf94325f43bb7} & GMX keeper price feed & liquidation & 24 & 0.58\% \\
\texttt{0x94c8821576fb9615e51b31ceec7648a09f94952c} & DODO pool admin (repricing) & DEX & 24 & 0.44\% \\
\texttt{0x1231deb6f5749ef6ce6943a275a1d3e7486f4eae} & LI.FI router & liquidation & 4 & 0.42\% \\
\texttt{0x98e6262eaf8956639a6b689360a0dc0a656e0229} & Gains Network oracle node & DEX + liquidation & 9 & 0.36\% \\
\texttt{0x9782d250657b0785c14c0ed56a2e30f971b13ca1} & Gains Network oracle node & DEX + liquidation & 7 & 0.34\% \\
\texttt{0x2d8f8b3b04849b1ee0aff0263b1e4d7f7980352a} & Gains Network oracle node & DEX + liquidation & 8 & 0.33\% \\
\texttt{0x1002ceb663d0e02e4945051555d6fd31ccd618e9} & Gains Network oracle node & DEX + liquidation & 8 & 0.32\% \\
\texttt{0xa5d07528a456a92a40091a4dbe42a898e3314dad} & Gains Network oracle node & DEX + liquidation & 7 & 0.31\% \\
\texttt{0x854d6022836b60c1758ea661a1b6f646d3e8a4e0} & Gains Network oracle node & DEX + liquidation & 7 & 0.31\% \\
\texttt{0xd276e9110bb87edc9124c754552317eaab986ad4} & Gains Network oracle node & DEX + liquidation & 6 & 0.30\% \\
\texttt{0x637025a9f5e2380e3bbe744ed4ffb3191cf5ef4d} & Gains Network oracle node & DEX + liquidation & 8 & 0.29\% \\
 & & & Total & 6.1\% \\
\bottomrule
\end{tabular}

  \caption{Flagged contracts excluded because they are not MEV bots. Each was identified from its verified source code, its on-chain links to its protocol, or its sampled transactions. Form is the probing form under which a contract was flagged, Months is the number of months in which it was flagged, and Share of spam gas is its share of the classifier's spam gas over Jan~2024 to Jun~2026.}
  \label{tab:excluded}
\end{table}

\section{MEV Revenue Scaling with User Volume}
\label{app:empirical-opportunity}

In this appendix, we estimate how the opportunity value $r$ scales with user demand in deployed systems.
This matters for the scaling analysis in the main body.
If MEV opportunity size grows with user activity, then the amount of spam that can be sustained at equilibrium may also grow with adoption.

\paragraph{Methodology.}
We study cyclic arbitrage on Ethereum Layer~2 rollups.
Our main case study is Base, and we also report corresponding estimates for Arbitrum.
We use daily observations and identify cyclic arbitrage transactions from decoded DEX swap data on Dune Analytics.
A transaction is classified as cyclic arbitrage if it contains at least three swaps that form a cycle, meaning that the first token sold is also the last token bought, and the transaction yields a positive net balance in exactly that token.
We exclude transactions routed through known DEX aggregator contracts, since these are more likely to reflect user-initiated routing rather than MEV.
For each arbitrage transaction, we compute gross profit from the net token balance and subtract the gas fee to obtain net profit.
We then aggregate these net profits by day and use the resulting daily total as a proxy for the available opportunity value $r$.
This proxy is imperfect, but it is natural in a competitive setting, since realized arbitrage profits should track the amount of extractable value available to searchers.

To measure user demand, we use daily non-MEV trading volume in USD and estimate the log-log regression
\begin{equation*}
\log(P_t)=\alpha+\gamma \log(V_t),
\end{equation*}
where $P_t$ is daily net cyclic arbitrage profit and $V_t$ is daily non-MEV trading volume.
The coefficient $\gamma$ is the elasticity of realized cyclic arbitrage profit with respect to user trading volume.

\paragraph{Base.}
We first analyze Base using 912 daily observations from January~2024 through June~2026.
The estimate is $\hat{\gamma}=1.04$ with Newey-West standard error $0.090$ (6 lags), $R^2=0.40$, and $p\approx 0$.
Thus, on Base, a $1\%$ increase in user trading volume is associated with roughly a $1.04\%$ increase in cyclic arbitrage profit, which is close to linear.

\paragraph{Arbitrum.}
We next apply the same methodology to Arbitrum.
For Arbitrum, we estimate $\hat{\gamma}=1.46$ with Newey-West standard error $0.120$ and $R^2=0.35$.
The point estimate exceeds one and is higher than on Base, indicating that realized cyclic arbitrage profit on Arbitrum grows super-linearly with user trading volume, with a fit comparable to Base.

\paragraph{Robustness.}
The elasticity $\gamma$ says by how many percent MEV profit rises when user activity rises by 1\%. We check whether it depends on how we measure user activity and on the period. First, we measure user activity by non-spam gas, the quantity in our model, instead of trading volume. On Arbitrum, this gives an elasticity of 1.73, close to the baseline of 1.46. On Base, it gives only 0.52, because Base's non-spam gas also rose and fell with its gas-target changes, which change capacity rather than demand. Second, we estimate the baseline regression separately for each half-year. The elasticity is positive in every half-year, but it ranges from 0.85 to 2.14 on Base and from 0.66 to 1.28 on Arbitrum (\cref{tab:mev-scaling-robustness}). MEV profit thus grows with user activity throughout, but how fast also depends on market conditions such as volatility. We therefore treat proportional growth of opportunity value as a modeling baseline that holds on average over our sample, not as a stable empirical regularity.

\begin{table}[h]
  \centering\footnotesize
  \begin{tabular}{lcccc}
    \toprule
    & \multicolumn{2}{c}{Base} & \multicolumn{2}{c}{Arbitrum} \\
    \cmidrule(lr){2-3}\cmidrule(lr){4-5}
    User activity measure & $\hat{\gamma}$ (SE) & $R^2$ & $\hat{\gamma}$ (SE) & $R^2$ \\
    \midrule
    Trading volume (baseline) & $1.04$ ($0.090$) & $0.40$ & $1.46$ ($0.120$) & $0.35$ \\
    Non-spam gas & $0.52$ ($0.126$) & $0.10$ & $1.73$ ($0.138$) & $0.31$ \\
    \midrule
    \multicolumn{5}{l}{\emph{Trading volume, estimated separately per half-year}} \\
    \quad 2024 H1 & $0.95$ ($0.140$) & $0.45$ & $1.20$ ($0.109$) & $0.47$ \\
    \quad 2024 H2 & $2.14$ ($0.133$) & $0.73$ & $0.66$ ($0.146$) & $0.21$ \\
    \quad 2025 H1 & $0.85$ ($0.297$) & $0.24$ & $0.91$ ($0.212$) & $0.18$ \\
    \quad 2025 H2 & $1.82$ ($0.160$) & $0.47$ & $0.74$ ($0.156$) & $0.13$ \\
    \quad 2026 H1 & $1.12$ ($0.334$) & $0.25$ & $1.28$ ($0.169$) & $0.33$ \\
    \bottomrule
  \end{tabular}
  \caption{Elasticity of daily cyclic-arbitrage profit with respect to user activity, i.e., the percentage increase in profit for a 1\% increase in user activity. Each row is a separate log-log regression of daily net profit on the stated measure of user activity, with Newey-West standard errors (6 lags) in parentheses. Trading volume is daily non-MEV DEX volume in USD. The full sample has 912 days per chain from Jan~2024 to Jun~2026, and each half-year about 182.}
  \label{tab:mev-scaling-robustness}
\end{table}

\end{document}